\pdfoutput=1
\documentclass[11pt]{article}

\usepackage{geometry}
\usepackage{times}
\usepackage{soul}
\usepackage{url}
\usepackage[hidelinks]{hyperref}
\usepackage[utf8]{inputenc}
\usepackage[small]{caption}
\usepackage{graphicx}
\usepackage{amsmath}
\usepackage{amsthm}
\usepackage{mathtools}
\usepackage{booktabs}
\usepackage{algorithm}
\usepackage{algorithmic}
\newtheorem{example}{Example}
\newtheorem{theorem}{Theorem}
\newtheorem{lemma}{Lemma}
\newtheorem{definition}{Definition}

\usepackage{natbib}
\usepackage{subcaption}
\usepackage{diagbox}
\usepackage{stmaryrd}
\usepackage{appendix}

\usepackage{dsfont}

\newcommand{\mat}[1]{\ensuremath{\textbf{#1}}}
\newcommand{\Nat}{\ensuremath{\mathds{N}}}
\newcommand{\Real}{\ensuremath{\mathds{R}}}
\newcommand{\bigO}[1]{\ensuremath{\mathcal{O}(#1)}}

\newcommand{\Vars}{\ensuremath{\mathcal{V}}}
\newcommand{\Dom}{\ensuremath{\Delta}}
\newcommand{\Preds}{\ensuremath{\mathcal{P}}}

\newcommand{\w}{\ensuremath{w, \overline{w}}}
\newcommand{\fotwo}{\ensuremath{\textbf{FO}^2}}
\newcommand{\ufotwo}{\ensuremath{\textbf{UFO}^2}}
\newcommand{\ccs}{\ensuremath{\fotwo+\textbf{CC}}}
\newcommand{\ufoccs}{\ensuremath{\ufotwo+\textbf{CC}}}
\newcommand{\ctwo}{\ensuremath{\textbf{C}^2}}

\title{Complexity of Weighted First-Order Model Counting in the Two-Variable Fragment with Counting Quantifiers:\\A Bound to Beat}

\date{%
Faculty of Electrical Engineering \\
Czech Technical University in Prague \\
Prague, Czech Republic \\%
}

\author{%
  Jan T\'{o}th
   \and
   Ond\v{r}ej Ku\v{z}elka
}

\begin{document}

\maketitle

\begin{abstract}
    We study the time complexity of weighted first-order model counting (WFOMC) over the logical language with two variables and counting quantifiers.
    The problem is known to be solvable in time polynomial in the domain size. However, the degree of the polynomial, which turns out to be relatively high for most practical applications, has never been properly addressed.
    First, we formulate a time complexity bound for the existing techniques for solving WFOMC with counting quantifiers.
    The bound is already known to be a polynomial with its degree depending on the number of cells of the input formula.
    We observe that the number of cells depends, in turn, exponentially on the parameters of the counting quantifiers appearing in the formula.
    Second, we propose a new approach to dealing with counting quantifiers, reducing the exponential dependency to a quadratic one, therefore obtaining a tighter upper bound.
    It remains an open question whether the dependency of the polynomial degree on the counting quantifiers can be reduced further, thus making our new bound a \emph{bound to beat}.
\end{abstract}

\section{Introduction}
The weighted first-order model counting (WFOMC) problem was originally proposed in the area of lifted inference as a method to perform probabilistic inference over statistical relational learning (SRL) models on the lifted level \citep{broeck-etal11:knowledge-compilation}.
It allowed, among other things, fast learning of various SRL models \citep{vanHaaren-etal15:lifted-learning}.
However, its applications have ranged beyond (symbolic) probabilistic reasoning since then, including conjecturing recursive formulas in enumerative combinatorics \citep{barvinek-etal21:p-recursions} and discovering combinatorial integer sequences \citep{svatos-etal23:fluffy}.

Regardless of the particular application context, WFOMC is also used to define a class of tractable (referred to as \emph{domain-liftable}) modeling languages, i.e., languages that permit WFOMC computation in time polynomial in the domain size. The logical fragment with two variables was the first to be identified as such \citep{broeck11:domain-liftability-fo2,broeck-etal14:wfomc-skolem}.
Negative result proving that logic with three variables contains $\verb|#|P_1$-complete counting problems followed \citep{beame-etal15:fo3-intractable}, spawning many attempts to recover at least some of the expressive power provided by three and more variables, yet retaining the domain-liftability property.

\citet{kazemi-etal16:sfo2+sru} introduced two new liftable classes, namely $\textbf{S}^2\fotwo$ and $\textbf{S}^2\textbf{RU}$.
\citet{kuusisto-lutz18:function-constraint} extended the two-variable fragment with one function constraint and showed such language to be domain-liftable, too.
That result was later generalized to the two-variable fragment with counting quantifiers, denoted by \ctwo{} \citep{kuzelka21:wfomc-in-c2}.
Moreover, several axioms can be added on top of the counting quantifiers, still retaining domain-liftability as well \citep{bremen-kuzelka21b:tree-axiom,toth-kuzelka23:linear-order,malhotra-serafini23:dag-axiom,malhotra-etal23:graph-axioms}. 

It follows from the domain-liftability of \ctwo{} that WFOMC computation time over formulas from \ctwo{} can be upper-bounded by a polynomial in the domain size.
However, it turns out that the polynomial's degree depends exponentially on the particular counting quantifiers appearing in the formula.
In this paper, we propose a new approach to dealing with counting quantifiers when computing WFOMC over \ctwo, which decreases the degree's dependency on the counting parameters from an exponential to a quadratic one, leading to a super-exponential speedup overall.

First, we review the necessary preliminaries in Section~\ref{sec:background}.
In Section~\ref{sec:old-ub}, we derive an upper bound on the time complexity of computing WFOMC over \ctwo{} for existing techniques.
In Section~\ref{sec:new-ub}, we continue by presenting a new technique for solving the problem, which improves the old upper bound considerably.
Section~\ref{sec:experiments} contains several experimental scenarios to support our theoretical results.
The runtime measurements therein show our improvement for several \ctwo{} sentences and also sentences from one of the domain-liftable \ctwo{} extensions.
Finally, we review related works in Section~\ref{sec:other-works} and conclude in Section~\ref{sec:conclusion}.

\section{Background}
\label{sec:background}
Notation-wise, we adhere to the standard way of writing both algebraic and logical formulas.
For readability purposes, we sometimes use $\cdot$ to denote multiplication, and other times, as is also common, we drop the operation sign.
We use boldface letters such as $\mat{x}$ to denote vectors and for any $n\in\Nat$, $\left[n\right]$ denotes the set $\{1, 2, \ldots, n\}$.

\subsection{First-Order Logic}
We work with a function-free subset of first-order logic (FOL).
A particular language is defined by a finite set of variables \Vars, a finite set of constants (also called the domain) \Dom\; and a finite set of predicates \Preds.
Assuming a predicate $P \in \Preds$ with arity $k$, we also write $P/k \in\Preds$.
An \emph{atom} has the form $P(t_1, t_2, \ldots, t_k)$ where $P/k \in \mathcal{P}$ and $t_i \in \Delta$ $\cup$ $\mathcal{V}$ are called \emph{terms}.
A \textit{literal} is an atom or its negation.
A set of \emph{formulas} can be then defined inductively.
Both atoms and literals are formulas.
Given some formulas, more complex formulas may be formed by combining them using logical connectives or by surrounding them with a universal ($\forall x$) or an existential ($\exists x$) quantifier where $x \in \Vars$.

A variable $x$ in a formula is called \emph{free} if the formula contains no quantification over $x$; otherwise, $x$ is called \emph{bound}.
A formula is called a \emph{sentence} if it contains no free variables.
A formula is called \emph{ground} if it contains no variables.

We use the definition of truth (i.e., semantics) from \emph{Herbrand Logic} \citep{hinrichs-geneserath06:hebrand-logic}.
A language's \emph{Herbrand Base} (HB) is the set of all ground atoms that can be constructed given the sets \Preds\; and \Dom.
A possible world, usually denoted by $\omega$, is any subset of HB.
Atoms contained in a possible world $\omega$ are considered to be true, the rest, i.e., those contained in $\text{HB} \setminus \omega$, are considered false.
The truth value of a more complex formula in a possible world is defined naturally.
A possible world $\omega$ is a \emph{model} of a formula $\varphi$ (denoted by $\omega \models \varphi$) if $\varphi$ is satisfied in $\omega$.

\subsubsection{FOL Fragments}
We often do not work with the entire language of FOL but rather some of its fragments. 
The simplest fragment we work with is the logic with at most two variables, denoted as \fotwo.
As the name suggests, any formula from \fotwo{} contains at most two logical variables.

The second fragment that we consider is the two-variable fragment with cardinality constraints.
We keep the restriction of at most two variables, but we recover some of the expressive power of logics with more variables by introducing an additional syntactic construct, namely a \emph{cardinality constraint} (CC).
We denote such language as \ccs.
CCs have the form $(|P| = k)$, where $P\in\Preds$ and $k\in\Nat$. 
Intuitively, such a cardinality constraint is satisfied in $\omega$ if there are exactly $k$ ground atoms with predicate $P$ therein.

Finally, the fragment that we pay most of our attention to is the two-variable logic with counting quantifiers, denoted by \ctwo{}.
A \emph{counting quantifier} is a generalization of the traditional existential quantifier.
For a variable $x \in \mathcal{V}$, a quantifier of the form $(\exists^{= k}x)$, where $k \in \Nat$, is allowed by our syntax.
Satisfaction of formulas with counting quantifiers is defined naturally, similar to the satisfaction of cardinality constraints.
For example, $\exists^{=k}x\;\psi(x)$ is satisfied in $\omega$ if there are exactly $k$ constants $\{A_1, A_2, \ldots, A_k\} \subseteq \Dom$ such that $\forall i \in [k]: \omega \models \psi(A_i)$.%
\footnote{Both CCs and counting quantifiers can also be defined to allow inequality operators. However, in the techniques that we study, inequalities are handled by transforming them to equalities \citep{kuzelka21:wfomc-in-c2}. After such transformation, one must repeatedly solve the case with equalities only. Hence, for brevity, we present only that one case.}

Note the distinction between cardinality constraints and counting quantifiers.
While the counting formula $\exists^{=k} x\; R(x)$ can be equivalently written using a single cardinality constraint $(|R|=k)$, the formula $\forall x \exists^{=k} y\; R(x, y)$ no longer permits such a simple transformation.

\subsection{Weighted First-Order Model Counting}
Let us start by formally defining the task that we study.
\begin{definition}{(Weighted First-Order Model Counting)}
    Let $\varphi$ be a formula over a fixed logical language $\mathcal{L}$.
    Let $(\w)$ be a pair of weight functions assigning a weight to each predicate in $\mathcal{L}$.
    Let $n$ be a natural number.
    Denote $\mathcal{MOD}(\varphi, n)$ the set of all models of $\varphi$ on a domain on size $n$.
    We define
\begin{align*}
    \textup{WFOMC}(\varphi, n, \w) = \sum_{\omega \in \mathcal{MOD}(\varphi, n)} \prod_{l \in \omega} w(\mathsf{pred}(l)) \prod_{l \in \textup{HB}\setminus\omega} \overline{w}(\mathsf{pred}(l)),
\end{align*}
where the function $\mathsf{pred}$ maps each literal to its predicate.
\end{definition}

In general, WFOMC is a difficult problem.
There exists a sentence with three logical variables, for which the computation is $\verb|#|P_1$-complete with respect to $n$ \citep{beame-etal15:fo3-intractable}.
However, for some logical languages, WFOMC can be computed in time polynomial in the domain size, which is also referred to as the language being \emph{domain-liftable}.

\begin{definition}{(Domain-Liftability)}
    Consider a logical language $\mathcal{L}$.
    The language is said to be \emph{domain-liftable} if and only if for any fixed $\varphi \in \mathcal{L}$ and any $n\in\Nat$, it holds that WFOMC($\varphi, n, \w$) can be computed in time polynomial in $n$.
\end{definition}
Thus, when we study domain-liftable languages, we focus on the time complexity with respect to the domain size, which is assumed to be the only varying input.
That is also called \emph{data complexity} of WFOMC in other literature \citep{beame-etal15:fo3-intractable}.

In the remainder of this text, it is often the case that we reduce one WFOMC computation (instance) to another over a different (larger) formula, possibly with fresh (not appearing in the original vocabulary) predicates.
Even then, the assumed input language remains fixed in the context of the new WFOMC instance.
Moreover, if we do not specify some of the weights for the new predicates, they are assumed to be equal to 1.

The first language proved to be domain-liftable was the language of \ufotwo, i.e., universally quantified \fotwo{} \citep{broeck11:domain-liftability-fo2}, which was later generalized to the entire \fotwo{} fragment \citep{broeck-etal14:wfomc-skolem}.
The original proof, making use of first-order knowledge compilation \citep{broeck-etal11:knowledge-compilation}, was later reformulated using \emph{1-types} (which we call \emph{cells}) from logic literature \citep{beame-etal15:fo3-intractable}.

\begin{definition}{(Valid Cell)}
A cell of a first-order formula $\varphi$ is a maximal consistent conjunction of literals formed from atoms in $\varphi$ using only a single variable.
Moreover, a cell $C(x)$ of a first-order formula $\varphi(x, y)$ is valid if and only if $\varphi(t, t) \wedge C(t)$ is satisfiable for any constant $t\in\Dom$.
\end{definition}

\begin{example}
Consider $\varphi = G(x) \vee H(x, y)$.

\noindent Then $\varphi$ has four cells:
\begin{align*}
    C_1(x) &= \neg G(x) \wedge \neg H(x, x), \\
    C_2(x) &= \neg G(x) \wedge H(x, x), \\
    C_3(x) &= G(x) \wedge \neg H(x, x), \\
    C_4(x) &= G(x) \wedge H(x, x).
\end{align*}

However, only cells $C_2, C_3$ and $C_4$ are valid.
\end{example}

Another domain-liftable language is \ccs.
WFOMC over \ccs{} is solved by repeated calls to an oracle solving WFOMC over \fotwo.
The number of required calls depends on the arities of predicates that appear in cardinality constraints.
Consider $\Upsilon=\bigwedge_{i=1}^m\left(|R_i|=k_i\right)$ and let us denote %
$\alpha(\Upsilon) = \sum_{i=1}^m {\left(arity(R_i)+1\right)}.$
For an \ccs{} formula $\Gamma = \Phi \wedge \Upsilon$ such that $\Phi\in\fotwo{}$ and $\Upsilon$ same as above, we will require $n^\alpha$ calls to the oracle \citep{kuzelka21:wfomc-in-c2}.

Having defined WFOMC, valid cells, and function $\alpha$, we can state known upper bounds for computing WFOMC over \fotwo{} and \ccs{}.
We concentrate those results into a single theorem.
\begin{theorem}
    \label{th:fo+cc-tractable}
    Let $\Gamma$ be an $\textup{\fotwo{}}$ sentence with $p$ valid cells.
    Let $\Upsilon=\bigwedge_{i=1}^m\left(|R_i|=k_i\right)$ be $m$ cardinality constraints, where $R_1, R_2, \ldots, R_m$ are some predicates from the language of $\Gamma$ and each $k_i\in\Nat$.
    For any $n\in\Nat$ and any fixed weights $(\w)$, \textup{WFOMC}($\Gamma, n, \w$) can be computed in time $\bigO{n^{p+1}}$, and
    \textup{WFOMC}($\Gamma \wedge \Upsilon, n, \w$) can be computed in time $\bigO{n^{\alpha(\Upsilon)}\cdot n^{p+1}}$.
    Since both the bounds are polynomials in $n$, both languages are domain-liftable.
\end{theorem}

The first bound follows from the cell-based domain-liftability proof \citep{beame-etal15:fo3-intractable}.%
{\footnote{The state-of-the-art algorithm for computing WFOMC over \fotwo{}, i.e., FastWFOMC, improves the bound considerably in some cases \citep{bremen-kuzelka21:fast-wfomc}. 
However, as the improvements are not guaranteed in the general case, we work with this bound as an effective worst case.
Moreover, as we demonstrate in the experimental section, our new encoding described further in the text improves the FastWFOMC runtime for \ctwo{} sentences reduced to \ccs{} as well.}}
The second bound follows from Propositions $4$ and $5$ in \citet{kuzelka21:wfomc-in-c2}.

It is important to note that Theorem \ref{th:fo+cc-tractable} assumes all mathematical operations to take constant time.
Hence, the theorem omits factors relating to \emph{bit complexity}, which \citet{kuzelka21:wfomc-in-c2} also addresses.
However, those factors remain the same in all transformations that we consider.
Therefore, we omit them for improved readability.
For a more detailed discussion on bit complexity, see the appendix.

\newpage
\subsection{Solving WFOMC with Counting Quantifiers}
Yet another domain-liftable language is the language of \ctwo.
WFOMC over \ctwo{} is solved by a reduction to WFOMC over \ccs{} \citep{kuzelka21:wfomc-in-c2}.
The reduction consists of several steps that we review in the lemmas below.
The lemmas concentrate results from other publications. Hence, we make appropriate references to each of them.

First, we review a specialized \emph{Skolemization} procedure for WFOMC \citep{broeck-etal14:wfomc-skolem}, which turns an arbitrary \fotwo{} sentence into a \ufotwo{} sentence.
Since all algorithms for solving WFOMC over \fotwo{} expect a universally quantified sentence as an input, this is a paramount procedure.
Compared to the source publication, we present a slightly modified Skolemization procedure.
The modification is due to \citet{beame-etal15:fo3-intractable}.
\begin{lemma}
\label{lem:skolem}
    Let $\Gamma=Q_1x_1Q_2x_2\ldots Q_kx_k \Phi(x_1,\ldots,x_k)$ be a~first-order sentence in prenex normal form with each quantifier $Q_i$ being either $\forall$ or $\exists$ and $\Phi$ being a quantifier-free formula.
    Denote by $j$ the first position of $\exists$.
    Let $\mat{x} = (x_1, \ldots, x_{j-1})$ and $\varphi(\mat{x}, x_j) = Q_{j+1} x_{j+1} \ldots Q_{k} x_k\Phi$.
    Set
    $$\Gamma' = \forall \mat{x} \left(\left(\exists x_j \varphi(\mat{x}, x_j)\right) \Rightarrow A(\mat{x})\right),$$
    where $A$ is a fresh predicate. 
    Then, for any $n\in\Nat$ and any weights $(\w)$ with $w(A) = 1$ and $\overline{w}(A)=-1$, it holds that
    $$\textup{WFOMC}(\Gamma, n, \w) = \textup{WFOMC}(\Gamma', n, \w).$$
\end{lemma}
Lemma \ref{lem:skolem} suggests how to eliminate one existential quantifier.
By transforming the implication inside $\Gamma'$ into a disjunction, we obtain a universally quantified sentence.
Repeating the procedure for each sentence in the input formula will eventually lead to one universally quantified sentence.

Next, we present a technique to eliminate negation of a subformula without distributing it inside.
The procedure builds on ideas from the Skolemization procedure, and it was presented as Lemma 3.4 in \citet{beame-etal15:fo3-intractable}.
It was also described as a \emph{relaxed Tseitin transform} in \citet{meert-etal16:relaxed-tseitin}.
\begin{lemma}
\label{lem:relaxed-tseitin}
    Let $\neg \psi(\mat{x})$ be a subformula of a first-order logic sentence $\Gamma$ with $k$ free variables $\mat{x} = (x_1, \ldots, x_k)$.
    Let $C/k$ and $D/k$ be two fresh predicates with $w(C)=\overline{w}(C)=w(D)=1$ and $\overline{w}(D)=-1$.
    Denote $\Gamma'$ the formula obtained from $\Gamma$ by replacing the subformula $\neg \psi(\mat{x})$ with $C(\mat{x})$.
    Let $$\Upsilon = \left(\forall \mat{x}\; C(\mat{x}) \vee D(\mat{x})\right) \wedge \left(\forall \mat{x}\; C(\mat{x}) \vee \psi(\mat{x})\right) \wedge \left(\forall \mat{x}\; D(\mat{x}) \vee \psi(\mat{x})\right).$$
    Then, it holds that
    $$\textup{WFOMC}(\Gamma, n, \w) = \textup{WFOMC}(\Gamma' \wedge \Upsilon, n, \w).$$
\end{lemma}

Finally, we move to dealing with counting quantifiers.%
\footnote{For brevity, the counting subformula in Lemmas \ref{lem:existsK}, \ref{lem:forall-existsK} and \ref{lem:A-or-existsK} contains only a single atom on a predicate $R$. That does not impede generality as the atom may represent a general subformula $\varphi$ equated to the atom using an additional universally quantified sentence, i.e., $\forall x \; R(x) \Leftrightarrow \varphi(x)$ or $\forall x \forall y\; R(x, y) \Leftrightarrow \varphi(x, y)$.}
We start with a single counting quantifier.
The approach follows from Lemma 3 in \citet{kuzelka21:wfomc-in-c2}.
\begin{lemma}
    \label{lem:existsK}
    Let $\Gamma$ be a first-order logic sentence.
    Let $\Psi$ be a \ctwo{} sentence such that $\Psi = \exists^{=k} x\; R(x)$.
    Let $\Psi' = (|R|=k)$ be a cardinality constraint.
    Then, it holds that
    \begin{align*}
        \textup{WFOMC}(\Gamma \wedge \Psi, n, \w) = \textup{WFOMC}(\Gamma \wedge \Psi', n, \w).
    \end{align*}
\end{lemma}

Next, we deal with a specific case of a formula quantified as $\forall\exists^{=k}$.
The following lemma was Lemma 2 in \citet{kuzelka21:wfomc-in-c2}
\begin{lemma}
\label{lem:forall-existsK}
    Let $\Gamma$ be a first-order logic sentence.
    Let $\Psi$ be a \ctwo{} sentence such that $\Psi = \forall x \exists^{=k} y\; R(x, y)$.
    Let $\Upsilon$ be an \textup{\ccs{}} sentence defined as
    \begin{align*}
        \Upsilon &=(|R|=k\cdot n) \wedge (\forall x \forall y\; R(x, y) \Leftrightarrow \bigvee_{i=1}^k f_i(x, y))\\
        &\wedge \bigwedge_{1\leq i<j\leq k} (\forall x \forall y\; f_i(x, y) \Rightarrow \neg f_j(x, y))\\
        &\wedge \bigwedge_{i=1}^k (\forall x \exists y\; f_i(x, y)),
    \end{align*}
    where $f_i/2$ are fresh predicates not appearing anywhere else.
    Then, it holds that
\begin{equation*}
    \textup{WFOMC}(\Gamma \wedge \Psi, n, \w) = \frac{1}{(k!)^n}\textup{WFOMC}(\Gamma \wedge \Upsilon, n, \w).
\end{equation*}%
\end{lemma}

Finally, we present a case that helps deal with an arbitrary counting formula.
It was originally presented as Lemma 4 in \citet{kuzelka21:wfomc-in-c2}.
\begin{lemma}
\label{lem:A-or-existsK}
    Let $\Gamma$ be a first-order logic sentence.
    Let $\Psi$ be a \textup{\ctwo{}} sentence such that $\Psi = \forall x\; A(x) \vee (\exists^{=k} y\; R(x, y))$.
    Define $\Upsilon = \Upsilon_1 \wedge \Upsilon_2 \wedge \Upsilon_3 \wedge \Upsilon_4$ such that
    \begin{align*}
        \Upsilon_1 &= \forall x \forall y\; \neg A(x) \Rightarrow (R(x, y) \Leftrightarrow B^R(x, y))\\
        \Upsilon_2 &= \forall x \forall y\; (A(x) \wedge B^R(x, y)) \Rightarrow U^R(y))\\
        \Upsilon_3 &= (|U^R|=k)\\
        \Upsilon_4 &= \forall x \exists^{=k} y\; B^R(x, y),
    \end{align*}
    where $U^R/1$ and $B^R/2$ are fresh predicates not appearing anywhere else.
    Then, it holds that
    $$\textup{WFOMC}(\Gamma \wedge \Psi, n, \w) = \frac{1}{\binom{n}{k}}\textup{WFOMC}(\Gamma \wedge \Upsilon, n, \w).$$
\end{lemma}

Algorithm~\ref{algo:c2-to-ccs} shows how to combine the lemmas above to reduce WFOMC over \ctwo{} to WFOMC over \ufoccs.

\begin{algorithm}[ht]
\caption{Converts \ctwo{} formulas into \ufoccs{}}
\label{algo:c2-to-ccs}
\textbf{Input}: Sentence $\Gamma \in \ctwo{}$\\
\textbf{Output}: Sentence $\Gamma^* \in \ufoccs{}$

\begin{algorithmic}[1] 

\FORALL{sentence $\exists^{=k} x\; \psi(x)$ in $\Gamma$}
    \STATE Apply Lemma 3
\ENDFOR

\FORALL{sentence $\forall x \exists^{=k} y\; \psi(x, y)$ in $\Gamma$}
    \STATE Apply Lemma 4
\ENDFOR

\FORALL{subformula $\varphi(x) = \exists^{=k} y\;\psi(x, y)$ in $\Gamma$}
    \STATE Create new predicates $R/2$ and $A/1$
    \STATE Let $\mu \gets \forall x \forall y\; R(x, y) \Leftrightarrow \psi(x, y)$
    \STATE Let $\nu \gets \forall x\; A(x) \Leftrightarrow (\exists^{=k} y\; R(x, y))$
    \STATE Apply Lemmas 2, 5, and 4 to $\nu$
    \STATE Replace $\varphi(x)$ by $A(x)$
    \STATE Append $\mu \wedge \nu$ to $\Gamma$
\ENDFOR

\FORALL{sentence with an existential quantifier in $\Gamma$}
    \STATE Apply Lemma 1
\ENDFOR

\RETURN $\Gamma$
\end{algorithmic}
\end{algorithm}

\newpage
\section{An Upper Bound for Existing Techniques}
\label{sec:old-ub}
As we have already mentioned above, WFOMC($\varphi, n, \w$) for $\varphi\in\ccs\;$ can be computed in time $\bigO{n^{p+1+\alpha}}$, where $p$ is the number of valid cells of $\varphi$ \citep{beame-etal15:fo3-intractable,kuzelka21:wfomc-in-c2} and $\alpha = \sum_{i=1}^m {\left(arity(R_i)+1\right)}$ with $R_1, R_2, \ldots, R_m$ being all the predicates appearing in cardinality constraints.
Let us see how the bound increases when computing WFOMC over \ctwo.

\subsection{A Worked Example on Removing Counting}
To be able to compute WFOMC of a particular \ctwo{} sentence, we must first encode the sentence in \ufoccs{}.
Let us start with an example of applying Algorithm~\ref{algo:c2-to-ccs} to do just that.

Consider computing WFOMC($\varphi, n, \w$) for the sentence
\begin{equation}
    \varphi = \exists^{=k} x \exists^{=l} y\; \psi(x, y),\label{eq:1}
\end{equation}
where $\psi$ is a quantifier-free formula from the two-variable fragment and $k, l \in \Nat$.

First, let us introduce two new fresh predicates, namely $R/2$ and $P/1$.
The predicate $R$ will \textit{replace} the formula $\psi$ and $P$ will do the same for the counting subformula $\exists^{=l} y\; \psi(x, y)$.
Specifically, we obtain
\begin{align}
    \varphi^{(1)} &= \left(\exists^{=k} x\; P(x)\right)\label{eq:2}\\
            &\wedge \left(\forall x\; P(x) \Leftrightarrow \left(\exists^{=l} y\; R(x, y)\right)\right).\label{eq:3}\\
            &\wedge \left(\forall x\forall y\; R(x, y) \Leftrightarrow \psi(x, y)\right)
\end{align}
While Sentence \ref{eq:2} is already easily encoded using a single cardinality constraint as Lemma \ref{lem:existsK} suggests, Sentence \ref{eq:3} requires more work.
Let us split the sentence into two implications:
\begin{align}
    \varphi^{(2)} &= \left(|P|=k\right)\wedge \left(\forall x\forall y\; R(x, y) \Leftrightarrow \psi(x, y)\right)\\
            &\wedge \left(\forall x\; P(x) \Rightarrow \left(\exists^{=l} y\; R(x, y)\right)\right)\label{eq:4}\\
            &\wedge \left(\forall x\; P(x) \Leftarrow \left(\exists^{=l} y\; R(x, y)\right)\right).\label{eq:5} 
\end{align}
Sentence $\ref{eq:4}$ can easily be rewritten into a form processable by Lemma \ref{lem:A-or-existsK}, whereas Sentence \ref{eq:5} will first need to be transformed using Lemma \ref{lem:relaxed-tseitin}, since the sentence can be rewritten~as
\begin{align*}
    \forall x\; P(x)\vee\neg\left(\exists^{=l} y\; R(x, y)\right).
\end{align*}
After applying Lemma \ref{lem:relaxed-tseitin}, we obtain
\begin{align}
    \varphi^{(3)} &= \left(|P|=k\right)\wedge \left(\forall x\forall y\; R(x, y) \Leftrightarrow \psi(x, y)\right)\\
        &\wedge \left(\forall x\; \overline{P}(x) \vee \left(\exists^{=l} y\; R(x, y)\right)\right)\label{eq:6}\\
        &\wedge \left(\forall x\; C(x) \vee \left(\exists^{=l} y\; R(x, y)\right)\right)\label{eq:7}\\
        &\wedge \left(\forall x\; D(x) \vee \left(\exists^{=l} y\; R(x, y)\right)\right)\label{eq:8}\\
        &\wedge \left(C(x)\vee D(x)\right) \wedge \left(\forall x\; \overline{P}(x) \Leftrightarrow \neg P(x)\right),
\end{align}
with $\overline{w}(D)=-1$.
Apart from applying Lemma~\ref{lem:relaxed-tseitin}, we also introduced another fresh predicate $\overline{P}/1$, which wraps the negation of $P$ for brevity further down the line.

Next, we need to apply Lemma \ref{lem:A-or-existsK} three times to Sentences \ref{eq:6} through \ref{eq:8}.
The repetitions can, luckily, be avoided.
By the distributive property of conjunctions and disjunctions, we can \emph{factor} $\overline{P}/1, C/1$ and $D/1$ out from the sentences, thus obtaining
\begin{align}
    \varphi^{(4)} &= \left(|P|=k\right)\wedge \left(\forall x\forall y\; R(x, y) \Leftrightarrow \psi(x, y)\right)\\
        &\wedge \left( \forall x \; F(x) \Leftrightarrow \left(C(x) \wedge D(x) \wedge \overline{P}(x)\right) \right) \\
        &\wedge \left(\forall x\; F(x) \vee \left(\exists^{=l} y\; R(x, y)\right)\right)\label{eq:678}\\
        &\wedge \left(C(x)\vee D(x)\right) \wedge \left(\forall x\; \overline{P}(x) \Leftrightarrow \neg P(x)\right),
\end{align}

Let us denote $\mathcal{T}_l(K, L)$ the transformation result when applying Lemma \ref{lem:A-or-existsK} to a sentence 
\begin{align*}
 \left(\forall x\; K(x) \vee \left(\exists^{=l} y\; L(x, y)\right)\right),
\end{align*}
which is now the form of Sentence~\ref{eq:678}.
After that application, our original sentence $\varphi$ becomes
\begin{align*}
    \varphi^{(5)} &= \left(|P|=k\right)\wedge \left(\forall x\forall y\; R(x, y) \Leftrightarrow \psi(x, y)\right)\\
        &\wedge \mathcal{T}_l(F, R) \wedge \left( \forall x \; F(x) \Leftrightarrow \left(C(x) \wedge D(x) \wedge \overline{P}(x)\right) \right)\\
        &\wedge \left(C(x)\vee D(x)\right) \wedge \left(\forall x\; \overline{P}(x) \Leftrightarrow \neg P(x)\right).
\end{align*}
We are still not done.
The result of Lemma \ref{lem:A-or-existsK} still contains another counting quantifier.
Specifically, the sentence $\mathcal{T}_l(K, L)$ contains the subformula
\begin{align*}
\forall x \exists^{=l} y\; B^L(x, y),    
\end{align*}
where $B^L/2$ is a fresh predicate.
We need Lemma~\ref{lem:forall-existsK} to eliminate this counting construct.
Using said lemma will turn $\mathcal{T}_l(K, L)$ into
\begin{align}
        \mathcal{T}'_l(K, L) &= \forall x \forall y\; \neg K(x) \Rightarrow \left(L(x, y) \Leftrightarrow B^L(x, y)\right)\\
        &\wedge \forall x \forall y\; (K(x) \wedge B^L(x, y)) \Rightarrow U^L(y))\\
        &\wedge \left(|U^L|=l\right) \wedge \left(|B^L|=n\cdot l\right)\\
        &\wedge \forall x \forall y\; B^L(x, y) \Leftrightarrow \bigvee_{i=1}^l f_i(x, y)\\
        &\wedge \bigwedge_{1\leq i<j\leq l} \left(\forall x \forall y\; f_i(x, y) \Rightarrow \neg f_j(x, y)\right)\\
        &\wedge \bigwedge_{i=1}^l \left(\forall x \exists y\; f_i(x, y)\right)\label{eq:9}.
\end{align}

We have obtained an \ccs{} sentence.
Unfortunately, we are still unable to directly compute WFOMC for such a formula either.
The problem lies in the $l$ sentences making up Formula~\ref{eq:9}, each containing an existential quantifier.
Following Lemma~\ref{lem:skolem}, we can replace each $\forall x \exists y\; f_i(x, y)$ with $\forall x \forall y\; \neg f_i(x, y) \vee A_i(x)$, where $A_i/1$ is a fresh predicate with $\overline{w}(A_i)=-1$ for each $i$.
Denote $\mathcal{T}''_l(K, L)$ the result of applying such change to $\mathcal{T}'_l(K, L)$.
Finally, we obtain a \ufoccs{} sentence
\begin{align*}
    \varphi^* &= \left(|P|=k\right)\wedge \left(\forall x\forall y\; R(x, y) \Leftrightarrow \psi(x, y)\right)\\
        &\wedge \mathcal{T}''_l(F, R) \wedge \left( \forall x \; F(x) \Leftrightarrow \left(C(x) \wedge D(x) \wedge \overline{P}(x)\right) \right)\\
        &\wedge \left(C(x)\vee D(x)\right) \wedge \left(\forall x\; \overline{P}(x) \Leftrightarrow \neg P(x)\right).
\end{align*}

\subsection{Deriving the Upper Bound}
As one can observe, the formula $\varphi^*$ in the example above has grown considerably compared to its original form.
The question is whether the formula growth can influence the asymptotic bound from Theorem~\ref{th:fo+cc-tractable}.
At first glance, the answer may seem negative.
That is due to the fact that the transformation only extends the vocabulary (which is assumed to be fixed), adds cardinality constraints (which are concentrated in the function $\alpha$), and increases the length of the input formula (which is also constant with respect to $n$).
However, there is one caveat to be aware of.
When extending the vocabulary, we may introduce new valid cells.
The vocabulary is fixed once WFOMC computation starts, but if we have formula $\Gamma\wedge\Phi$ such that $\Gamma\in\ccs{}$ and $\Phi\in\ctwo{}$, then this formula already has $p$ valid cells.
Once we construct $\Gamma\wedge\Phi^*$, where $\Phi^*\in\ccs{}$ is obtained from $\Phi$ using Algorithm~\ref{algo:c2-to-ccs}, the new formula will have $p^*$ valid cells and, possibly, $p\leq p^*$.
If we wish to express a complexity bound for \ctwo, we should inspect the possible increase in $p$ to obtain $p^*$.

To deal with an arbitrary \ctwo{} formula, we need to be able to deal with subformulas such as the one in Sentence~\ref{eq:3}.
As one can observe from both the example above and Algorithm~\ref{algo:c2-to-ccs}, encoding such a sentence in \ufoccs{} requires, in order of appearance, Lemmas~\ref{lem:relaxed-tseitin}, \ref{lem:A-or-existsK}, \ref{lem:forall-existsK} and \ref{lem:skolem}.
Applying Lemmas~\ref{lem:relaxed-tseitin} and \ref{lem:A-or-existsK} introduces only a constant number of fresh predicates.
Hence, the increase in $p$ can be expressed by multiplying with a constant $\beta$.
See the proof of Theorem~\ref{th:c2-bound} for the derivation of a value for $\beta$.

Although the constant $\beta$ may increase our polynomial degree considerably, there is another, much more substantial, influence.
An application of Lemma~\ref{lem:forall-existsK}, which additionally requires Lemma~\ref{lem:skolem} to deal with \emph{unskolemized} formulas such as in Formula~\ref{eq:9} will introduce $2k$ new predicates.
Although $k$ is a parameter of the counting quantifiers, i.e., part of the language, and the language is assumed to be fixed, the size of the encoding of \ctwo{} in \ufoccs{} obviously depends on $k$.
Hence, the number of cells may also increase with respect to $k$, and, as we formally state below, it, in fact, does.

\begin{lemma}
    \label{lem:bound-forall-existsK}
    For any $m\in\Nat$, there exists a sentence $\Gamma = \varphi \wedge \bigwedge_{i=1}^m \left(\forall x \exists^{=k_i} y\; \psi_i(x, y)\right)$ such that the \ufoccs{} encoding of $\Gamma$ obtained using Lemma~\ref{lem:forall-existsK} has $\bigO{p \cdot \prod_{i=1}^m \gamma(k_i)}$ valid cells, where $p$ is the number of valid cells of $\varphi$ and $\gamma(k) = (k + 2) \cdot 2^{k-1}$.
\end{lemma}
\begin{proof}
    Consider the sentence 
    \begin{align}
        \Gamma = &\bigwedge_{i=1}^m \left( \forall x \forall y\; E_i(x, y) \Rightarrow E_i(y, x) \right) \wedge\label{eq:19}\\
        &\bigwedge_{i=1}^m \left( \forall x \exists^{=k_i} y\; E_i(x, y)\right).\label{eq:20}
    \end{align}
    In this setting, $\varphi$ is Formula~\ref{eq:19} and it has $p = 2^m$ valid cells.
    We need to apply Lemmas~\ref{lem:forall-existsK} and \ref{lem:skolem} $m$ times to encode sentences in Formula~\ref{eq:20} into \ufoccs.
    Let us investigate one such application.

    First, consider valid cells of $\varphi$, that contain $E_i(x, x)$ negatively.
    Then all $f_{ij}(x, x)$ must also be negative (the index $j$ now refers to the predicates introduced in a single application of Lemma~\ref{lem:forall-existsK}), which will immediately satisfy all \emph{skolemization clauses} obtained by application of Lemma~\ref{lem:skolem}.
    Hence, the atoms $A_{ij}(x)$ will be allowed to be present either positively or negatively for all $j$.
    Thus, the number of such cells will increase $2^{k_i}$ times.

    Second, consider valid cells of $\varphi$, that contain $E_i(x, x)$ positively.
    Then exactly one of $f_{ij}(x, x)$ can be satisfied in each cell.
    That will cause the number of cells to be multiplied by $k_i$.
    Next, for a particular cell, denote $t$ the index such that $f_{it}(x, x)$ is positive in that cell.
    Then all $A_{ij}(x)$ such that $j \neq t$ will again be free to assume either a positive or a negative form.
    Only $A_{it}(x)$ will be fixed to being positive.
    Hence, the number of such cells will further be multiplied by $2^{k_i-1}$.

    Overall, for a single application of Lemma~\ref{lem:forall-existsK}, the number of cells will be $\bigO{p \cdot 2^{k_i} + p \cdot k_i \cdot 2^{k_i-1}}\in\bigO{p \cdot  (k_i+2)\cdot2^{k_i-1}}$, where we upper bound both partitions of the valid cells of $\varphi$ by their total number.
    After repeated application of Lemmas \ref{lem:forall-existsK} and \ref{lem:skolem}, the bound above directly leads to the bound we sought to prove.%
    \footnote{In the proof, we opted for as simple formula as possible. It would be easy to handle computing WFOMC for $\Gamma$ by decomposing the problem into $m$ identical and independent problems. It is, however, also easy to envision a case where such decomposition is not as trivial. Consider adding constraints such that for each $x$ and $y$, there is only one $i$ such that $E_i(x, y)$ is satisfied.}
\end{proof}

One last consideration for dealing with all of \ctwo{} is Lemma~\ref{lem:existsK}.
The lemma adds one new cardinality constraint to the formula, which does not increase the number of cells in any way.
It will, however, require more calls to an oracle for WFOMC over \fotwo.
That influence on the overall bound can still be concentrated in the function $\alpha$.
To distinguish the values from before encoding \ctwo{} and after, let us denote the new value $\alpha'$.

\newpage
Finally, we are ready to state the overall time complexity bound for computing WFOMC over \ctwo{}.
\begin{theorem}
    \label{th:c2-bound}
    Consider an arbitrary $\ctwo{}$ sentence rewritten as 
    $$\varphi=\Gamma\wedge\bigwedge_{i=1}^{m}\left(\forall x\; P_i(x) \Leftrightarrow \left(\exists^{=k_i} y\; R_i(x, y)\right)\right),$$ where
    $\Gamma\in\ccs$.
    For any $n\in\Nat$ and any fixed weights $(\w)$, $\textup{WFOMC}(\varphi, n, \w)$ can be computed in time $\bigO{n^{\alpha'}\cdot n^{1 + p \cdot \prod_{i=1}^m \beta \cdot \gamma(k_i)}}$, where $p$ is the number of valid cells of $\Gamma$, $\alpha'$ and $\beta$ are constants with respect to both $n$ and the counting parameters $k_i$ and $\gamma(k) = (k + 2) \cdot 2^{k-1}$.
\end{theorem}
\begin{proof}
Suppose that
\begin{align*}
    \Gamma &= \psi \wedge \bigwedge_{i=1}^{m_1} (|Q_i| = l_{i})
\end{align*}
with $\psi\in\fotwo{}$ and $Q_i$ being predicates from the vocabulary of $\psi$.
Consider encoding a single \ctwo{} sentence
\begin{align*}
    \mu_i = \left(\forall x\; P_i(x) \Leftrightarrow \left(\exists^{=k_i} y\; R_i(x, y)\right)\right)
\end{align*}
into \ccs{}.

First, following Algorithm~1 and the example in Section~3.1, we apply Lemma~2 and leverage distributive law, obtaining
\begin{align*}
    \mu^{(1)}_i &= \left(\forall x\; C(x) \vee D(x)\right) \wedge \left(\forall x\; \overline{P_i}(x) \Leftrightarrow \neg P_i(x)\right)\\
    &\wedge \left(\forall x\; A(x) \Leftrightarrow \left(C(x) \wedge D(x) \wedge \overline{P_i}\right)\right)\\
    &\wedge \left(\forall x\; A(x) \vee \left(\exists^{=k_i} y\; R_i(x, y)\right)\right),
\end{align*}
with $C/1, D/1$ and $A/1$ being fresh predicates and $\overline{w}(D)=-1$.
Second, using Lemma~5, we construct
\begin{align*}
    \mu^{(2)}_i &= \left(\forall x\; C(x) \vee D(x)\right) \wedge \left(\forall x\; \overline{P_i}(x) \Leftrightarrow \neg P_i(x)\right)\\
    &\wedge \left(\forall x\; A(x) \Leftrightarrow \left(C(x) \wedge D(x) \wedge \overline{P_i}\right)\right)\\
    &\wedge \left(\forall x\forall y\; \neg A(x) \Rightarrow \left(R_i(x, y) \Leftrightarrow B^{R_i}(x, y)\right)\right)\\
    &\wedge \left(\forall x\forall y\; \left( A(x) \wedge B^{R_i}(x, y)\right) \Rightarrow U^{R_i}(y)\right)\\
    &\wedge \left(|U^{R_i}|=k_i\right) \wedge \left(\forall x\exists^{=k_i} y\; B^{R_i}(x, y)\right),
\end{align*}
where $U^{R_i}/1$ and $B^{R_i}/2$ are also fresh predicates.

Let us now inspect how the number of cells may have increased so far.
If the original formula had $p$ valid cells, then after application of Lemma~2, there may be up to $4p$ cells.
To show that, consider partitioning valid cells into those that satisfy the subformula $\psi(\mat{x})$ considered by Lemma~2 and those that do not.
In the former case, the new cells must satisfy both $C(x)$ and $D(x)$. Thus, the number of cells remains unchanged.
In the latter case, $C(x)$ or $D(x)$ must be satisfied, leading to a factor of three.
Hence, we have an upper bound of $p + 3p = 4p$.

Similarly, for Lemma~5, in cells that do not satisfy $A(x)$, $B^R(x, x)$ is determined, but $U^R(x)$ is unconstrained, leading to a factor of two.
In the other cells, if $B^R(x, x)$ is also satisfied, then $U^R(x)$ is determined. However, it is unconstrained if $B^R(x, x)$ is satisfied, which can be upper bounded by considering only the second case, leading to another factor of two.
Hence, we have another factor of $2+2 = 4$.
Overall, in the worst case, the increase in the number of cells by Lemmas~2 and 5 is by a factor of
\begin{align*}
    \beta = 8,
\end{align*}
which will be applied each time we process one sentence $\mu_i$.

Next, let us investigate applying Lemmas~4 and 1 to $\mu^{(2)}_i$ which produces

\begin{align*}
    \mu_i^{(3)} &= \left(\forall x\; C(x) \vee D(x)\right) \wedge \left(\forall x\; \overline{P_i}(x) \Leftrightarrow \neg P_i(x)\right)\\
    &\wedge \left(\forall x\; A(x) \Leftrightarrow \left(C(x) \wedge D(x) \wedge \overline{P_i}\right)\right)\\
    &\wedge \left(\forall x\forall y\; \neg A(x) \Rightarrow \left(R_i(x, y) \Leftrightarrow B^{R_i}(x, y)\right)\right)\\
    &\wedge \left(\forall x\forall y\; \left( A(x) \wedge B^{R_i}(x, y)\right) \Rightarrow U^{R_i}(y)\right)\\
    &\wedge \left(|U^{R_i}|=k_i\right) \wedge \left(|B^{R_i}| = n\cdot k_i\right)\\
    &\wedge \left(\forall x \forall y\; B^L(x, y) \Leftrightarrow \bigvee_{i=1}^{k_i} f_i(x, y)\right)\\
    &\wedge \bigwedge_{1\leq i<j\leq k_i} \left(\forall x \forall y\; f_i(x, y) \Rightarrow \neg f_j(x, y)\right)\\
    &\wedge \bigwedge_{i=1}^l \left(\forall x \forall y\; \neg f_i(x, y) \vee A_i(x)\right),
\end{align*}
where $f_i/2$ and $A_i/1$ are fresh predicates with $\overline{w}(A_i)=-1$.
It follows from Lemma~6 that applying Lemmas 4 and 1 may increase the number of valid cells by a factor of
\begin{align*}
    \gamma(k_i) = (k_i + 2) \cdot 2^{k_i-1}.
\end{align*}
The factor is an upper bound because the sentence $\Gamma$ shown in the proof of Lemma~6 affords the atoms with the predicates $f_i$ and $A_i$ the highest possible number of degrees of freedom.
The truth values of the atoms are determined only by sentences added through the application of Lemmas ~4 and 1.
Therefore, there cannot be more valid cells.

Finally, assuming that each $\mu_i$ is encoded into \ccs{} independently of others, we can substitute all values into the bound from Theorem 1, obtaining the bound from our claim.

One thing remaining is to evaluate the factor related to the cardinality constraints in our final formula and make sure it is a constant with respect to both $n$ and the counting parameters.
Following \cite{kuzelka21:wfomc-in-c2}, the factor is
\begin{align*}
    \alpha' &= \left(\sum_{i=1}^m arity(U^{R_i}) + arity(B^{R_i})\right)\\
    &\hspace{1cm} + \left( \sum_{i=1}^{m_1} arity(Q_i) + 1\right)\\
    &= 3m + m_1 + \sum_{i=1}^{m_1} arity(Q_i).
\end{align*}

\end{proof}

\section{Improving the Upper Bound}
\label{sec:new-ub}
Let us now inspect the bound from Theorem~\ref{th:c2-bound}.
Although it is polynomial in $n$, meaning that $\ctwo{}$ is, in fact, domain-liftable \citep{kuzelka21:wfomc-in-c2}, we can see that the number of valid cells (a part of the polynomial's degree) grows exponentially with respect to the counting parameters $k_i$.
In this section, we propose an improved encoding to the one from Lemma~\ref{lem:forall-existsK} which reduces said growth to a quadratic one.

The new encoding does not build on entirely new principles, instead, it takes the existing transformation and makes it more efficient.
As in Lemma~\ref{lem:forall-existsK}, we will describe the situation for dealing with one $\forall\exists^{=k}$-quantified subformula.
The procedure could easily be generalized to having $m\in\Nat$ such subformulas by repeating the process for each of them independently.

The most significant issue with the current encoding are the \emph{Skolemization} predicates $A_i/1$, which increase the number of valid cells exponentially with respect to $k$.
The new encoding will seek to constrain those predicates so that the increase is reduced.
Let us start with a formula
\begin{align}
    \Gamma = \varphi \wedge \forall x \exists^{=k} y\; \psi(x, y),\label{eq:forall-existsK}
\end{align}
where $\varphi\in\ccs{}$ and $\psi$ is quantifier-free.
Let us also consider the encoding of $\Gamma$ in $\ufoccs$, i.e.,
\begin{align*}
    \Gamma^* &= \varphi \wedge \left(\forall x \forall y\; R(x, y) \Leftrightarrow \psi(x, y)\right) \wedge (|R| = n \cdot k) \\
             &\wedge \left( \forall x \forall y\; R(x, y) \Leftrightarrow \bigvee_{i=1}^k f_i(x, y) \right)\\
             &\wedge \bigwedge_{1\leq i<j\leq k} \left(\forall x \forall y\; f_i(x, y) \Rightarrow \neg f_j(x, y)\right)\\
            &\wedge \bigwedge_{i=1}^k \left(\forall x \forall y\; \neg f_i(x, y) \vee A_i(x)\right),
\end{align*}
with fresh predicates $R/2, f_i/2$ and $A_i/1$ and weights $\overline{w}(A_i) = -1$ for all $i \in \left[k\right]$.

\subsection{Canonical Models}
The new encoding will leverage a concept that we call a \emph{canonical model}, which we gradually build in this subsection.

Let $\omega$ be a model of $\Gamma^*$ and $t\in\Dom$ be an arbitrary domain element.
Denote $\mathcal{A}^t \subseteq \left[k\right]$ the set of indices such that 
\begin{align*}
    \omega \models \bigwedge_{j\in\mathcal{A}^t} A_{j}(t) \wedge \bigwedge_{j\in\left[k\right]\setminus\mathcal{A}^t} \neg A_{j}(t).
\end{align*}
Now, let us transform $\omega$ into $\omega_t$, which will be another model of $\Gamma^*$.

First, we separate all atoms in $\omega$ (atoms true in $\omega$) without the predicates $f_i/2$ and $A_i/1$ into the set $\mathcal{R}_0$, atoms on $A_i/1$ not containing the constant $t$ into $\mathcal{R}_t^A$ and atoms on $f_i/2$ not containing the constant $t$ on the first position into $\mathcal{R}_t^f$.

Second, we define an auxiliary injective function $g_t: \mathcal{A}^t \mapsto \left[k\right]$ mapping elements of $\mathcal{A}^t$ to the first $|\mathcal{A}^t|$ positive integers, i.e.,
\begin{align}
\label{def:g}
    g_t(j) = | \{j' \in \mathcal{A}^t\; |\; j' \leq j \} |
\end{align}
Third, we define a set of atoms $\mathcal{A}^{new}_t$ such that
\begin{align*}
    \mathcal{A}^{new}_t = \{ A_{{g_t(j)}}(t)\; |\; \omega \models A_{j}(t)\},
\end{align*}
i.e., we accumulate the Skolemization atoms with the constant $t$ that are satisfied in $\omega$ and we change their indices to the first $|\mathcal{A}^t|$ positive integers.
Next, we do a similar thing for atoms with $f_i$'s and the constant $t$ at the first position.
Note that we use the same function $g_t$ that was defined (with respect to $\mathcal{A}^t$) in Equation~\ref{def:g}.
Hence, we construct a set
\begin{align*}
    \mathcal{F}^{new}_t = \{ f_{{g_t(j)}}(t,t') \; |\; \omega \models f_{j}(t, t'), t' \in \Dom\}.
\end{align*}
Finally, we are ready to define the new model of $\Gamma^*$ as
\begin{align*}
    \omega_t = \mathcal{R}_0 \cup \mathcal{A}^{new}_t \cup \mathcal{F}^{new}_t \cup \bigcup_{t' \in \Delta\setminus\{t\}} \left(\mathcal{R}_{t'}^A \cup \mathcal{R}_{t'}^f\right).
\end{align*}

\begin{lemma}
\label{lem:omega-t}
    For any $\omega \models \Gamma^*$ and any $t\in\Dom$, $\omega_t$ constructed as described above is another model of $\Gamma^*$.
\end{lemma}
\begin{proof}
    To prove the claim, it is sufficient to note that we can permute the indices $i$, and the sentence $\Gamma^*$ will remain the same.
    The transformation permutes the indices so that the atoms $A_i(t)$ that are satisfied have the lowest possible indices.
\end{proof}

Now, suppose that we take some model of $\Gamma^*$ and we repeatedly perform the transformation described above for all domain elements, i.e., for some ordering of the domain such as $\Dom=\{t_1, t_2, \ldots, t_n\}$, we construct $\omega_{t_1}$ from $\omega$, then we construct $\omega_{t_2}$ from $\omega_{t_1}$ and so on, until we obtain $\omega_{t_n}=\omega^*$.
Note that several models $\omega$ can lead to the same $\omega^*$. Thus, $\omega^*$ effectively induces an equivalence class. 

\begin{definition}
    A model $\omega^*$ constructed in the way described above is a \emph{canonical model} of $\Gamma^*$.
    Moreover, all models of $\Gamma^*$ that lead to the same canonical model are called \emph{A-equivalent}.\footnote{The letter ``A'' simply refers to the Skolemization predicates that we call $A_i$, although they could be called anything else.}
\end{definition}

\noindent A property of A-equivalent models will be useful in what follows.
We formalize it as another lemma.

\begin{lemma}
\label{lem:canon-models}
    Let $\omega^*$ be a canonical model of $\Gamma^*$.
    There are
    \begin{align}
        \prod_{t\in\Dom} \binom{k}{|\mathcal{A}^t|}\label{eq:canon_models}
    \end{align}
    models that are A-equivalent to $\omega^*$.
    Moreover, any two A-equivalent models have the same weight.
\end{lemma}
\begin{proof}
    It follows from the proof of Lemma~\ref{lem:omega-t} that the atoms of $A_i$'s that are true in $\omega^*$ have the lowest possible indices.
    Hence, for a fixed $t\in\Dom$, the number of models that lead to $\omega^*$ depends only on the number of ways that we can split the indices between the satisfied and the unsatisfied atoms.
    There are $k$ indices to choose from, and for a fixed assignment, the set $\mathcal{A}^t$ holds the indices of satisfied $A_j(t)$'s, which leads directly to Equation~\ref{eq:canon_models}.

    The second claim follows from the fact that the transformation of any $\omega$ into $\omega^*$ does not change the number of true atoms of any given predicate.
\end{proof}

\subsection{The New Encoding}
In this subsection, we use the concept of canonical models and observations from Lemma~\ref{lem:canon-models} to devise an encoding that counts only canonical models.
By weighing them accordingly, we then recover the correct weighted model count of the original problem.

Let us start by introducing new fresh predicates $C_i/1$.
We will want an atom $C_j(t)$ to be true if and only if $A_1(t), A_2(t), \ldots A_j(t)$ were the only Skolemization atoms satisfied in a model of $\Gamma^*$ (the old encoding).
Intuitively, the predicates $C_i$ will constrain the models to only correspond to canonical models.
Thus, we define $C_i$'s as
\begin{align*}
    \Gamma_C = \bigwedge_{j=0}^{k} &\forall x\; C_{j}(x) \Leftrightarrow
    \left( \bigwedge_{h\in\left[j\right]} A_{h}(x) \wedge \bigwedge_{h\in\left[k\right]\setminus\left[j\right]} \neg A_{h}(x) \right).
\end{align*}
The new encoding can then be described as
\begin{align}
    \Gamma^{new} = \Gamma^* \wedge \Gamma_C \wedge \left( \forall x\; \bigvee_{j=0}^{k} C_{j}(x) \right).\label{eq:new-encoding}
\end{align}
The final disjunction was added to make sure that we only count canonical models (at least one of $C_i$'s is satisfied).

One more thing to consider is the weights for the predicates $C_i$.
A particular atom $C_j(t)$ is satisfied in a model $\omega$ if the model is a canonical model corresponding to $C_j(t)$ according to the sentence $\Gamma_C$.
Following Lemma~\ref{lem:canon-models}, such a model represents an entire set of A-equivalent models, each with the same weight.
Hence, if $C_j(t)$ is satisfied, we should count the model weight of $\omega$ as many times as how many A-equivalent models to $\omega$ there are.
If $C_j(t)$ is, on the other hand, unsatisfied, we want to keep the weight the same.
Therefore, we set $w(C_i)=\binom{k}{i}$ and $\overline{w}(C_i)=1$ for all $i\in\left[k\right]$.
\begin{lemma}
For any sentence $\Gamma\in\ctwo{}$ such as the one in Equation~\ref{eq:forall-existsK} and $\Gamma^{new}\in\ufoccs{}$ obtained from $\Gamma$ as in Equation~\ref{eq:new-encoding}, for any $n\in\Nat$ and any weights $(\w)$ extended for predicates $C_i$ as above, it holds that
    \begin{align*}
    \textup{WFOMC}(\Gamma,n,\w) = \frac{1}{(k!)^n}\textup{WFOMC}(\Gamma^{new},n,\w). 
    \end{align*}
\end{lemma}
\begin{proof}
$\Gamma_C$ encodes each canonical model of $\Gamma^*$ using a single $C_j$ predicate indicating that for any $t\in\Dom$, $A_1(t), A_2(t), \ldots, A_j(t)$ are satisfied.
For all $i\in\{0,1,\ldots,k\}$, $\Gamma_C$ allows at most one $C_i(t)$ to be true for any $t$.
The disjunction in $\Gamma^{new}$ requires at least one $C_i(t)$ to be true.
Combined, each model of $\Gamma^{new}$ is one of the canonical models of $\Gamma^*$.
Hence, we are only counting canonical models of $\Gamma^*$.

It follows from Lemma~8 that for a canonical model $\omega$ represented by $C_i$, there are $\binom{k}{i}$ models of $\Gamma^*$ that are A-equivalent to $\omega$.
Hence, when we have one canonical model, it represents $\binom{k}{i}$ models that must all be considered to obtain the correct weighted model count, leading to the weight $w(C_i)=\binom{k}{i}$ for each $C_i$.

So far, we have shown that we are counting models of $\Gamma^*$.
However, we must still deal with the overcounting introduced by applying Lemma~4 to turn $\Gamma$ into $\Gamma^*$.
Therefore, we use the same factor as in Lemma~4 to only count models of $\Gamma$.
\end{proof}

\subsection{The Improved Upper Bound}
With the new encoding, we can decrease the number of valid cells of a sentence obtained after applying Lemma~\ref{lem:forall-existsK} to a \ctwo{} sentence.
Hence, we can improve the upper bound on time complexity of computing WFOMC over \ctwo{}.
Using the same notation as in Theorem~\ref{th:c2-bound}, we can formulate Theorem~\ref{th:new-c2-bound}.
\begin{theorem}
    \label{th:new-c2-bound}
    Consider an arbitrary $\ctwo{}$ sentence rewritten as $\varphi=\Gamma\wedge\Phi$, where
    $\Gamma\in\ccs$ and
    $$\Phi = \bigwedge_{i=1}^{m}\left(\forall x\; P_i(x) \Leftrightarrow \left(\exists^{=k_i} y\; R_i(x, y)\right)\right).$$
    For any $n\in\Nat$ and any fixed weights $(\w)$, $\textup{WFOMC}(\varphi, n, \w)$ can be computed in time $$\bigO{n^{\alpha'}\cdot n^{1 + p \cdot \prod_{i=1}^m \beta \cdot \gamma'(k_i)}},$$ where $p$ is the number of valid cells of $\Gamma$ and $\gamma'(k) = \bigO{k^2 + 2k + 1}$.
\end{theorem}
\begin{proof}
    Let us derive $\gamma'$. All other values can be derived identically as in the proof of Theorem~\ref{th:c2-bound}.

    Let us return to the sentence from the proof of Lemma~\ref{lem:bound-forall-existsK}, i.e, $$\Gamma = \bigwedge_{i=1}^m \left( \forall x \forall y\; E_i(x, y) \Rightarrow E_i(y, x) \right) \wedge
        \bigwedge_{i=1}^m \left( \forall x \exists^{=k_i} E_i(x, y)\right).$$
    As we already know, sentence $\Gamma$ causes the largest increase in the number of valid cells.
    
    With the old encoding, most of the truth values of the Skolemization atoms with predicates $A_i$ were unconstrained, leading to an exponential blowup.
    Due to Equation~\ref{eq:new-encoding}, that is no longer the case.
    The sentence $\Gamma_C$ forces $C_j(t)$ to be true if and only if $A_1(t), A_2(t), \ldots, A_j(t)$ are true and all other $A_{i}(t)$ with $j < i \leq k$ are false.
    Hence, if $C_{j'}(t)$ is true, then all other $C_j(t)$ with $j' \neq j$ are false.
    Moreover, at least one $C_i(t)$ must be true due to the final disjunction in Equation~\ref{eq:new-encoding}.

    Therefore, we only have $(k+1)$ possibilities for assigning truth values to atoms with $C_i$ predicates.
    The truth values of atoms with $A_i$'s are then directly determined without any more degrees of freedom.
    Hence, for both cases considered in the proof of Lemma~\ref{lem:bound-forall-existsK}, we receive a factor $(k+1)$ instead of the exponential.
    Thus, we obtain $$\gamma' = \bigO{p \cdot (k+1) + p \cdot k \cdot (k+1)} \in \bigO{p \cdot (k^2 + 2k + 1)}.$$
\end{proof}

\section{Experiments}
\label{sec:experiments}
In this section, we support our theoretical findings by providing time measurements for various WFOMC computations using both the old and the new \ctwo{} encoding.
We also provide tables comparing the number of valid cells $p$, which determines the polynomial degree, giving a more concrete idea of the speedup provided by the new encoding.
Last, but not least, we provide an experiment performing lifted inference over a Markov Logic Network \citep{richardson-domingos06:mln} defined using the language of \ctwo{} extended with the linear order axiom.

For all of our experiments, we used \emph{FastWFOMC.jl},\footnote{\url{https://github.com/jan-toth/FastWFOMC.jl}} an open source Julia implementation of the FastWFOMC algorithm \citep{bremen-kuzelka21:fast-wfomc}, which is arguably the state-of-the-art, reported by its authors to outperform the first approach to computing WFOMC in a lifted manner, i.e, ForcLIFT,%
\footnote{\url{https://dtaid.cs.kuleuven.be/wfomc}}
which is based on knowledge compilation \citep{broeck-etal11:knowledge-compilation}.
Apart from time measurements for \ctwo{} sentences, we also inspect sentences from one of the domain-liftable \ctwo{} extensions, namely \ctwo{} with the linear order axiom \citep{toth-kuzelka23:linear-order}.

Most of our experiments were performed in a single thread on a computer with an AMD Ryzen 5 7500F CPU running at 3.4GHz and having 32 GB RAM.
Problems containing the linear order axiom, which have considerably higher memory requirements, were solved using a machine with AMD EPYC 7742 CPU running at 2.25GHz and having 512 GB of RAM.

\subsection{Performance Measurements}
In this section, we present observed running times for several problems specified using the language of \ctwo{} or its extension by the linear order axiom.
Most of our experiments show simply counting the number of graphs with some specified properties.
The advantage of using counting problems is that we can easily check the obtained results using the On-Line Encyclopedia of Integer Sequences \citep{oeis}.
To perform counting as opposed to weighted counting, we simply set the weights (both positive and negative) of all predicates in the input formula to one.

\subsubsection{Counting $k$-regular Graphs}
First, consider a \ctwo{} sentence encoding $k$-regular undirected graphs without loops, i.e.,
\begin{align*}
    \Gamma_1 &= \left(\forall x\; \neg E(x, x)\right) \wedge \left(\forall x \forall y\; E(x, y) \Rightarrow E(y, x)\right)\\ 
    &\wedge \left(\forall x \exists^{=k} y\; E(x, y)\right).
\end{align*}

Figure~\ref{fig:k-regs} shows the measured running times of $\textup{WFOMC}(\Gamma_1, n, 1, 1)$ for $k\in\{3, 4, 5\}$.
As one can observe, the new encoding surpasses the old in each case.
The difference may not seem as distinct for $k=3$ compared to $k=4$ or $k=5$, but it is still substantial.
For one, runtime for $n=51$ already exceeded runtime of $1000$ seconds in the case of the old encoding, whereas the new encoding did not reach that value even for $n=70$.
Additionally to the figure, Table~\ref{tab:k-regs} shows the number of valid cells $p_{old}$ produced by the old encoding of $\Gamma_1$ into \ufoccs{} and $p_{new}$ produced by the new one.
As one can observe, e.g., for $k=5$, the new encoding reduces the runtime from $\bigO{n^{33}}$ to $\bigO{n^{7}}$.

\vfill

\begin{table}[ht]  
    \centering
    \begin{tabular}{c|c|c|c}
         $k$ & 3 & 4 & 5 \\\hline
         $p_{old}$ & 8 & 16 & 32\\
         $p_{new}$ & 4 & 5 & 6
    \end{tabular}
    \caption{Number of valid cells for $k$-regular graphs}
    \label{tab:k-regs}
\end{table}

\vfill

\begin{figure*}[t]
     \centering
     \begin{subfigure}[b]{0.32\linewidth}
         \centering
         \includegraphics[width=\textwidth]{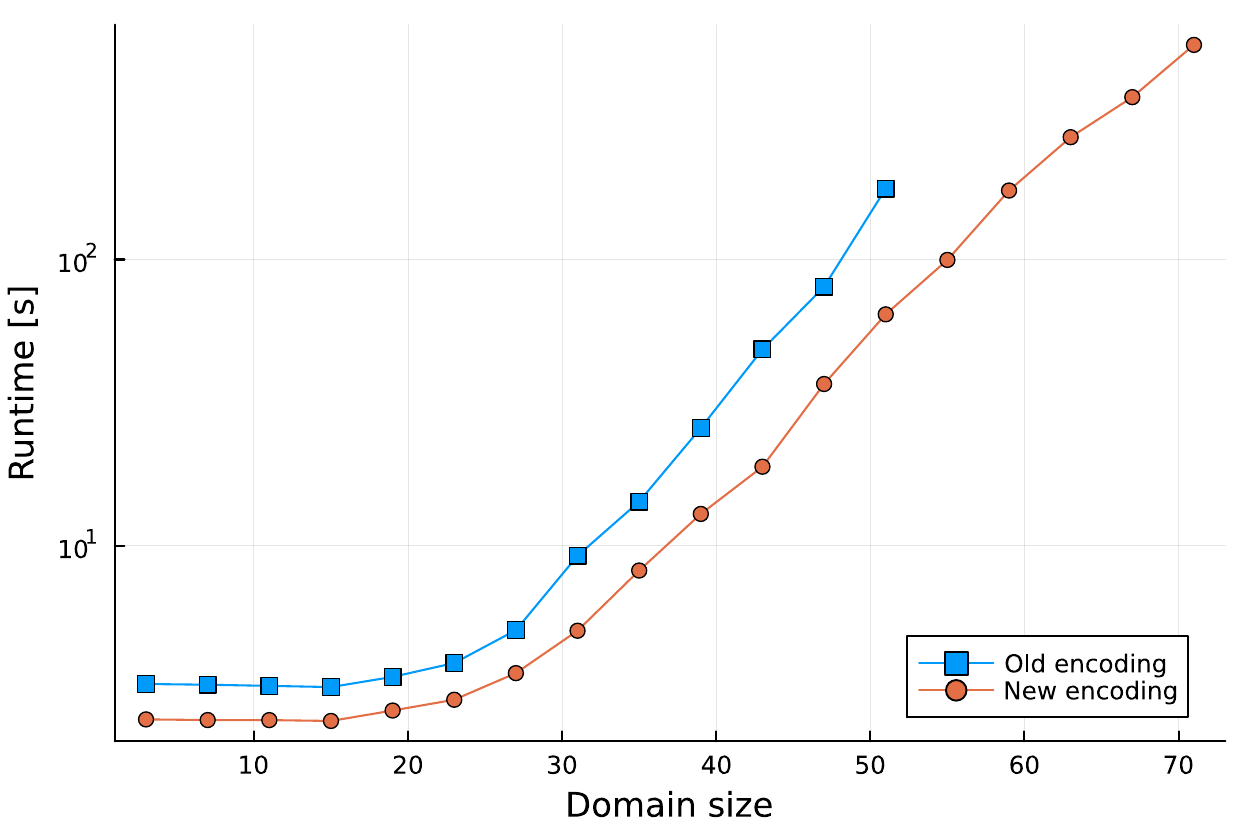}
         \caption{$3$-regular}
         \label{fig:3-reg}
     \end{subfigure}
     \begin{subfigure}[b]{0.32\linewidth}
         \centering
         \includegraphics[width=\textwidth]{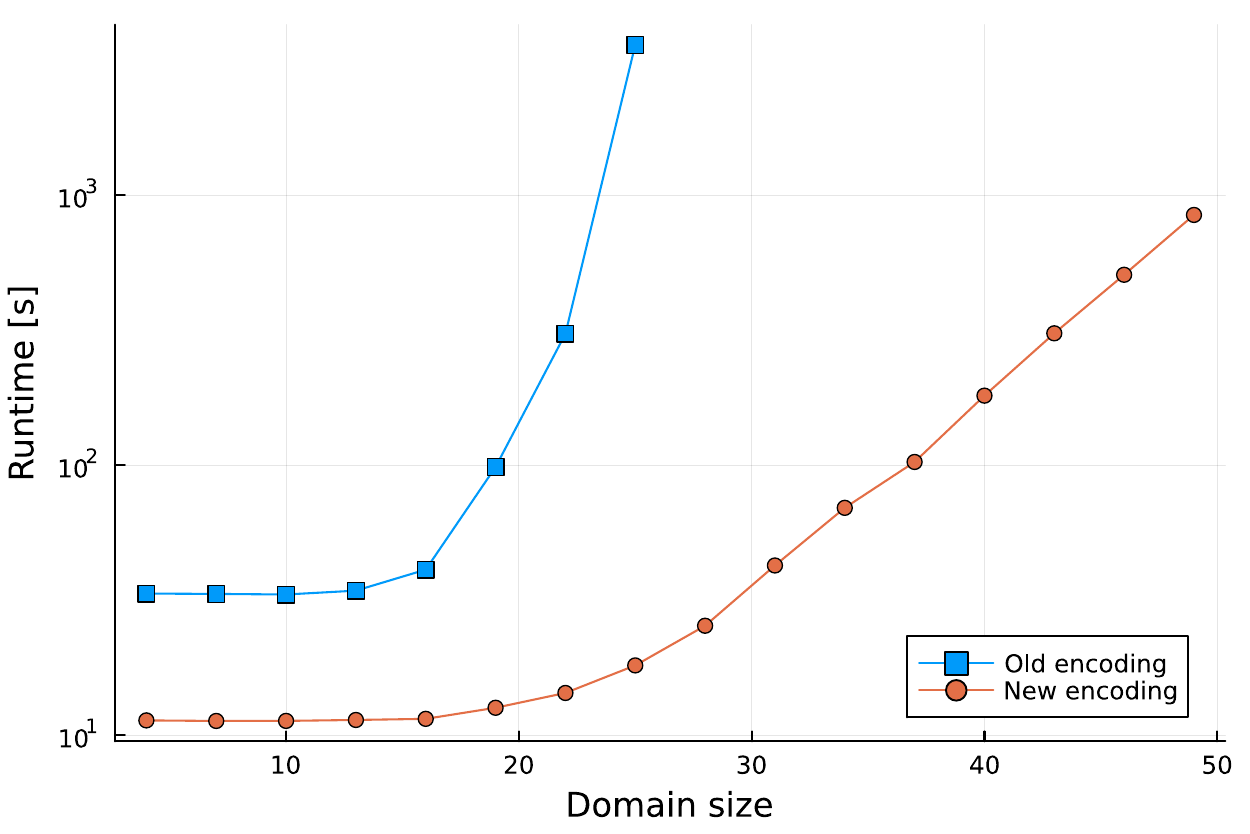}
         \caption{$4$-regular}
         \label{fig:4-reg}
     \end{subfigure}
     \begin{subfigure}[b]{0.32\linewidth}
         \centering
         \includegraphics[width=\textwidth]{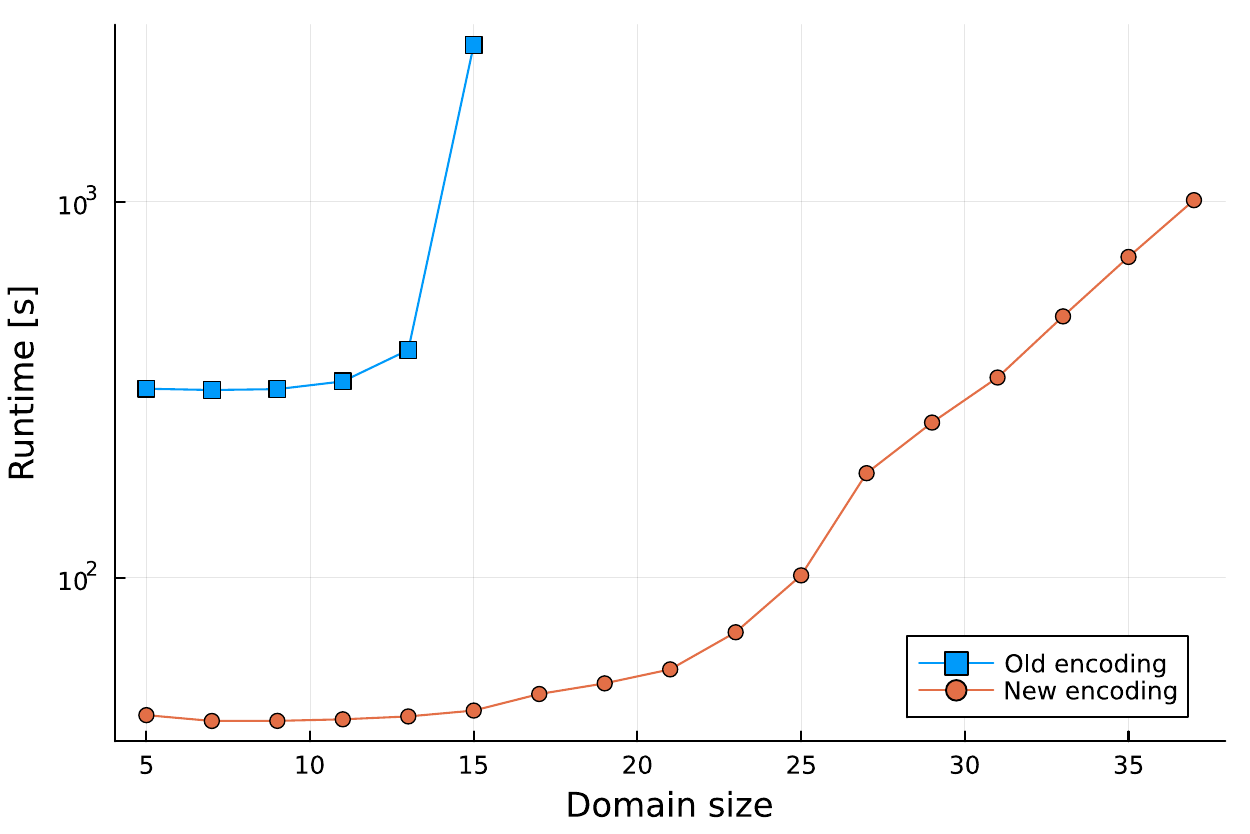}
         \caption{$5$-regular}
         \label{fig:5-reg}
     \end{subfigure}
    \caption{Runtime for counting $k$-regular graphs}
    \label{fig:k-regs}
\end{figure*}

\newpage
\subsubsection{Counting $k$-regular $l$-colored Graphs}
Next, we can extend the $k$-regular graph definition to also include a coloring by $l$ colors, effectively obtaining a logical description of $k$-regular $l$-colored graphs.%
\footnote{Note that $l$-colored graphs are graphs along with a coloring using at most $l$ colors.
That is different from $l$-colorable graphs which is a set of graphs that can be colored using $l$ colors.}
Sentence $\Gamma_2$ encodes such graphs.
\begin{align*}
    \Gamma_2 &= \left(\forall x\; \neg E(x, x)\right) \wedge \left(\forall x \forall y\; E(x, y) \Rightarrow E(y, x)\right)\\ 
    &\wedge \left(\forall x \exists^{=k} y\; E(x, y)\right)\\
    &\wedge \left(\forall x\; \left(\bigvee_{i=1}^l C_i(x)\right)\right) \wedge \bigwedge_{1\leq i < j \leq l} \left(\forall x\; C_i(x) \Rightarrow \neg C_j(x)\right)\\ 
    &\wedge \left(\forall x \forall y\; E(x, y) \Rightarrow \left(\bigwedge_{i=1}^l \neg \left(C_i(x) \wedge C_i(y)\right)\right)\right)\\
\end{align*}

The runtime measurements of $\textup{WFOMC}(\Gamma_2, n, 1, 1)$ are available in Figure~\ref{fig:kreg-lcol}.
Similarly to the previous case, Table~\ref{tab:kreg-lcols} shows the number of valid cells for various scenarios, even for some that are not included in Figure~\ref{fig:kreg-lcol} due to their computational demands.

\vfill

\begin{table}[ht]
    \hfill
    \begin{subtable}[ht]{0.45\linewidth}
    \begin{tabular}{c|c|c|c}
         \diagbox{$l$}{$k$} & 3 & 4 & 5 \\\hline
            2 & 16 & 32 & 64\\
            3 & 24 & 48 & 96\\
            4 & 32 & 64 & 128\\
    \end{tabular}
    \caption{Old encoding}
    \end{subtable}
    \hfill
    \begin{subtable}[ht]{0.45\linewidth}
    \begin{tabular}{c|c|c|c}
         \diagbox{$l$}{$k$} & 3 & 4 & 5 \\\hline
            2 & 8 & 10 & 12\\
            3 & 12 & 15 & 18\\
            4 & 16 & 20 & 24\\
    \end{tabular}
    \caption{New encoding}
    \end{subtable}
    \hfill
    \caption{Number of valid cells for $k$-regular $l$-colored graphs}
    \label{tab:kreg-lcols}
\end{table}

\vfill

\begin{figure}[ht]
     \centering
     \begin{subfigure}[b]{0.32\linewidth}
         \centering
         \includegraphics[width=\textwidth]{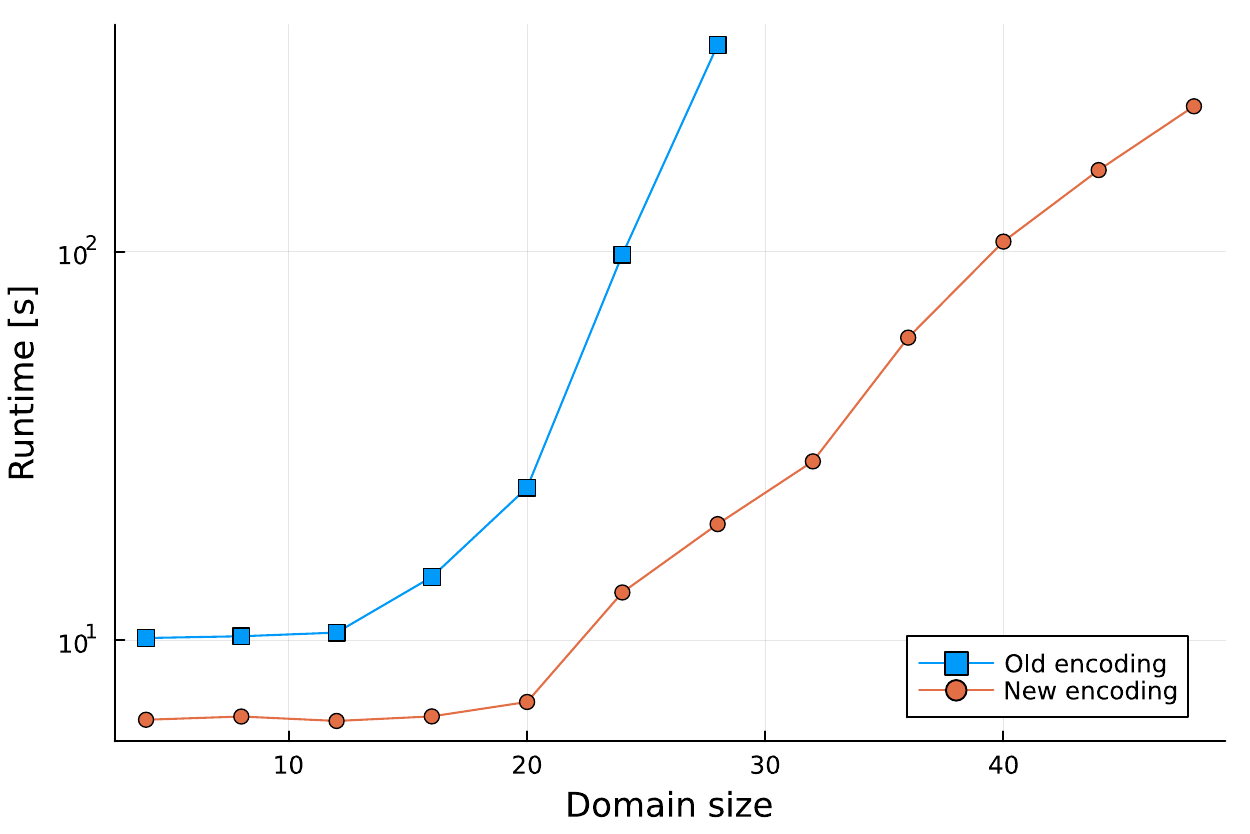}
         \caption{$3$-regular $2$-colored}
         \label{fig:3reg-2col}
     \end{subfigure}
     \begin{subfigure}[b]{0.32\linewidth}
         \centering
         \includegraphics[width=\textwidth]{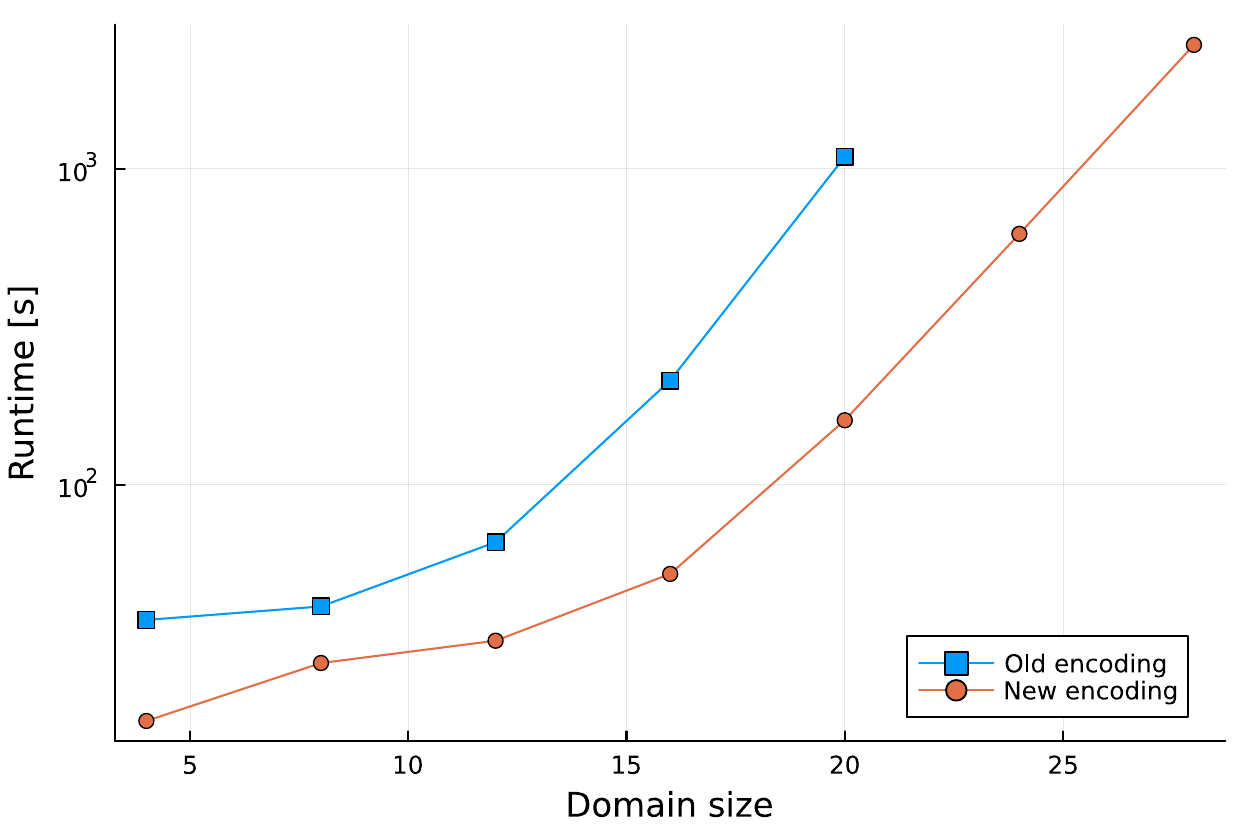}
         \caption{$3$-regular $3$-colored}
         \label{fig:3reg-3col}
     \end{subfigure}
     \begin{subfigure}[b]{0.32\linewidth}
         \centering
         \includegraphics[width=\textwidth]{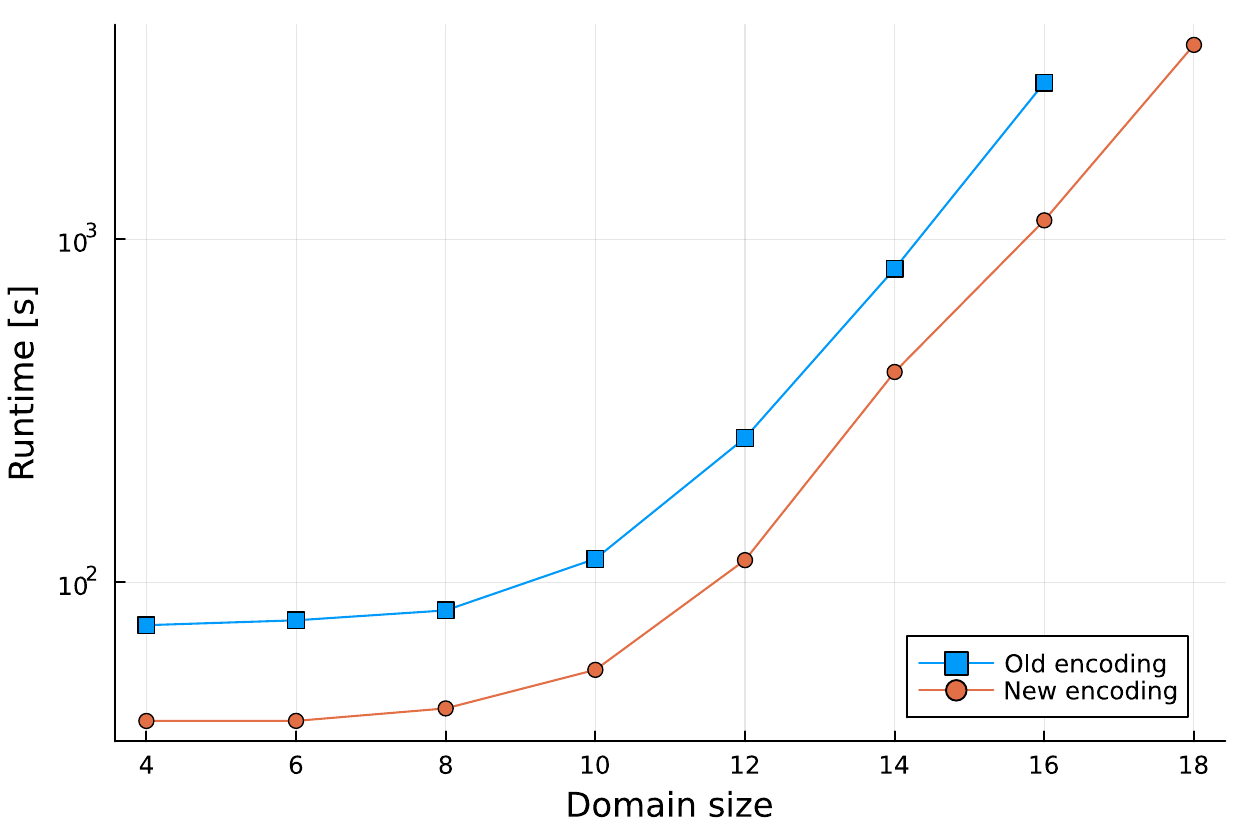}
         \caption{$3$-regular $4$-colored}
         \label{fig:3reg-4col}
     \end{subfigure}

     \begin{subfigure}[b]{0.32\linewidth}
         \centering
         \includegraphics[width=\textwidth]{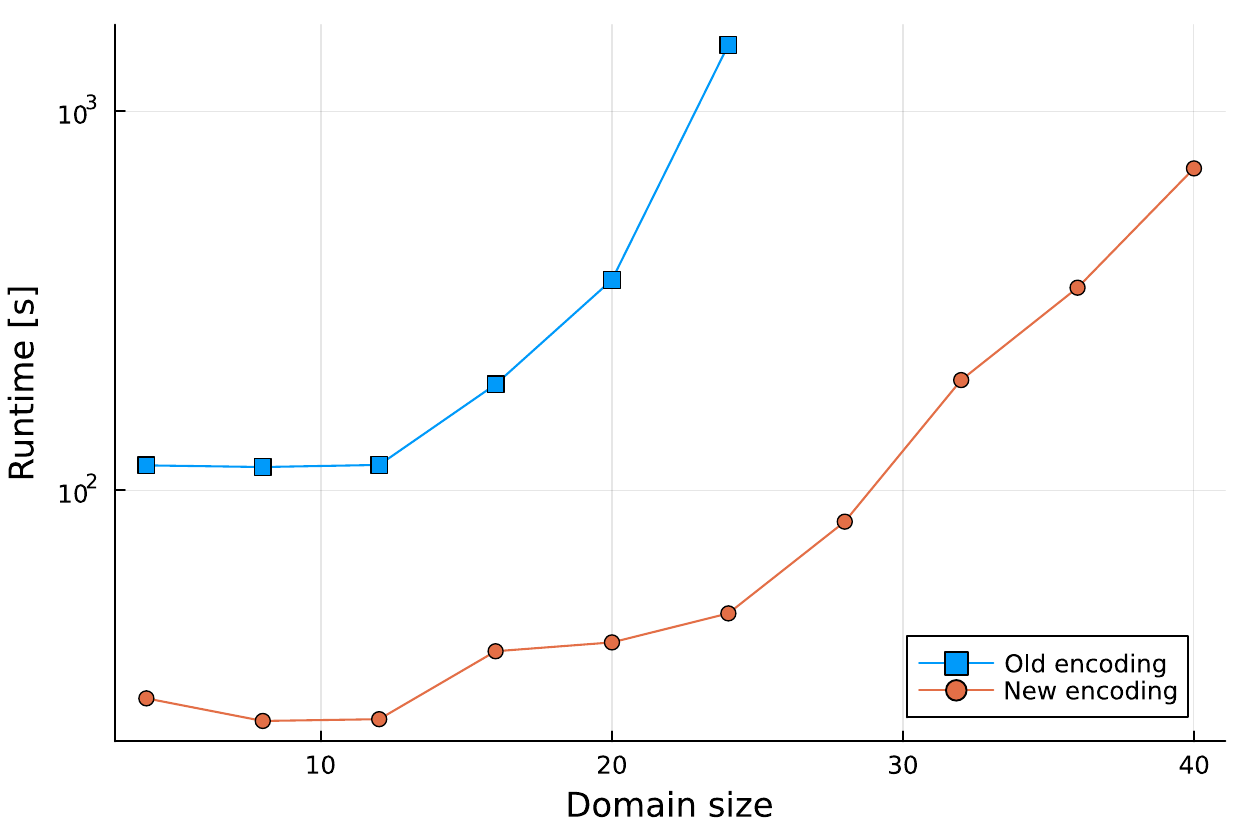}
         \caption{$4$-regular $2$-colored}
         \label{fig:4reg-2col}
     \end{subfigure}
     \begin{subfigure}[b]{0.32\linewidth}
         \centering
         \includegraphics[width=\textwidth]{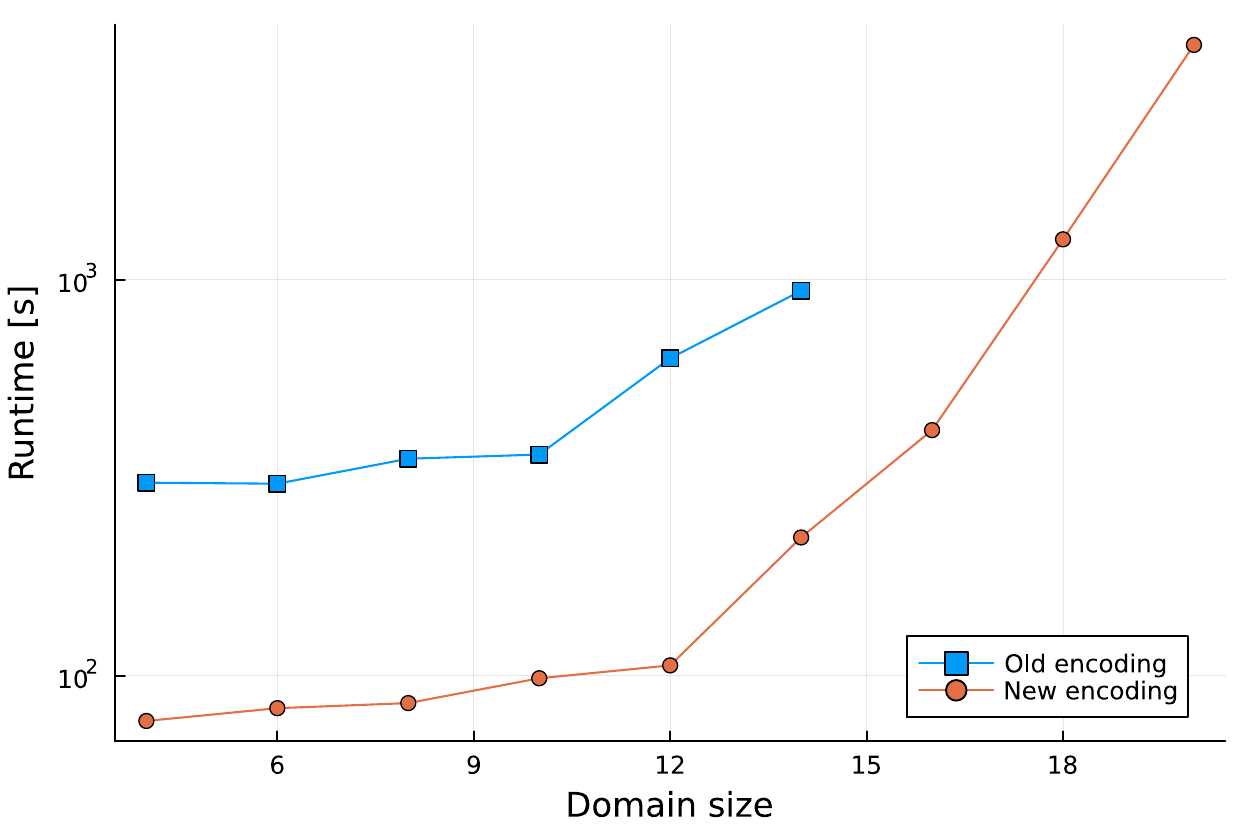}
         \caption{$4$-regular $3$-colored}
         \label{fig:4reg-3col}
     \end{subfigure}
     \begin{subfigure}[b]{0.32\linewidth}
         \centering
         \includegraphics[width=\textwidth]{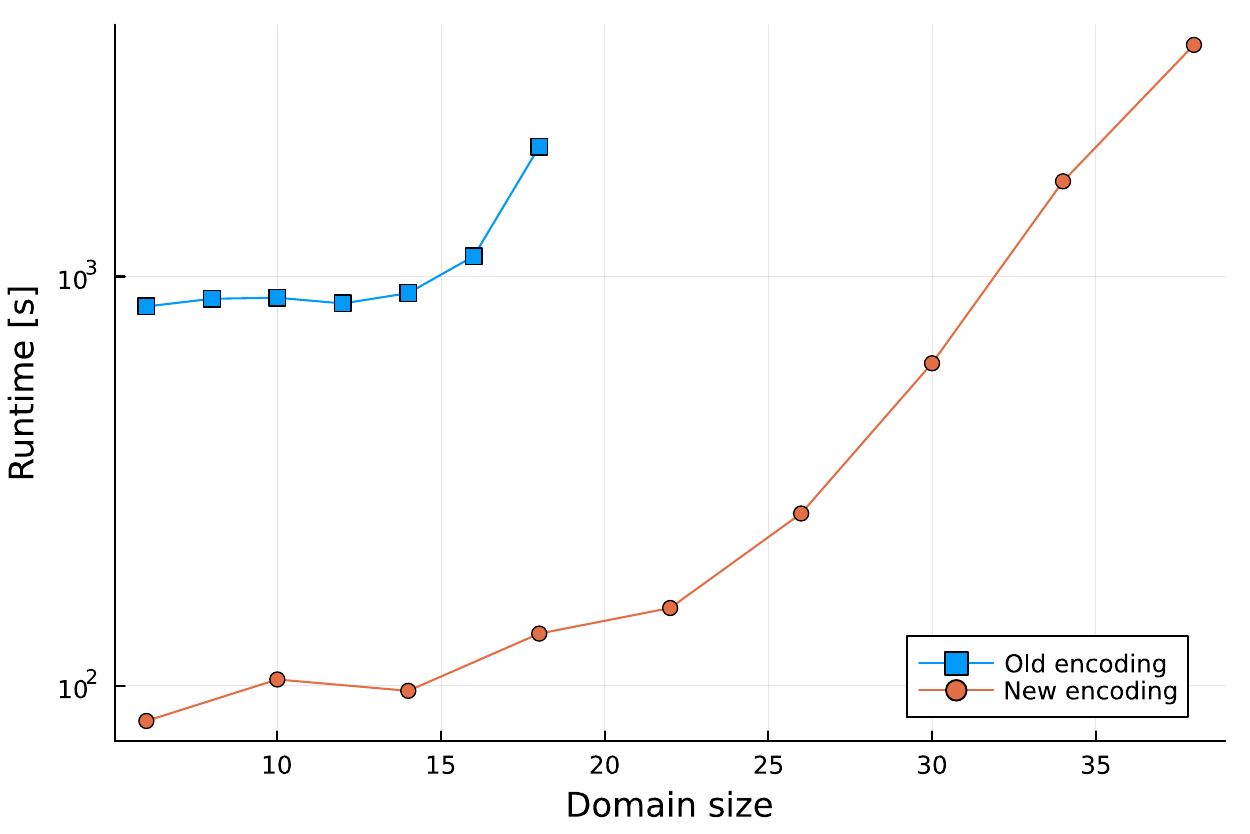}
         \caption{$5$-regular $2$-colored}
         \label{fig:5reg-2col}
     \end{subfigure}
    \caption{Runtime for counting $k$-regular $l$-colored graphs}
    \label{fig:kreg-lcol}
\end{figure}

\newpage
\subsubsection{Counting $k$-regular Directed Graphs}
In the examples so far, our \ctwo{} sentences contained only one $\forall\exists^{=k}$-quantified sentence.
Let us now see a case with more than one counting sentence, which should put an even larger emphasis on our upper bounds.
Directed $k$-(bi)regular graphs without loops require both the number of outgoing and incoming edges to be exactly $k$ for each vertex.
We can encode them using $\Gamma_3$, which uses two counting sentences:
\begin{align*}
    \Gamma_3 &= \left(\forall x\; \neg E(x, x)\right)\\ &\wedge \left(\forall x \exists^{=k} y\; E(x, y)\right)\\ &\wedge \left(\forall x \exists^{=k} y\; E(y, x)\right)
\end{align*}

Runtimes for counting over such digraphs may be inspected in Figure~\ref{fig:3reg-digraph}.
Table~\ref{tab:kreg-digraphs} depicts valid cell counts, which turn out to be quite high.
Due to those high values, which end up being polynomial degrees in our upper bounds, the problem is considerably more difficult.
Even for $k = 3$, the old encoding was extremely inefficient. Hence, we only provide runtimes for counting with the new encoding in that case.%
\footnote{Note that for the case of $k=3$, the corresponding OEIS sequence at \url{https://oeis.org/A007105} also contains only elements up to $n = 14$, which is the same as we managed to compute within a time limit of $10^5$ seconds.}

\vfill

\begin{table}[tbh]
    \centering
    \begin{tabular}{c|c|c|c|c}
         $k$ & 2 & 3 & 4 & 5 \\\hline
         $p_{old}$ & 16 & 64 & 256 & 1024\\
         $p_{new}$ & 9 & 16 & 25 & 36
    \end{tabular}
    \caption{Number of valid cells for $k$-regular digraphs}
    \label{tab:kreg-digraphs}
\end{table}

\vfill

\begin{figure*}[tbh]
    \hfill
    \begin{subfigure}{0.45\linewidth}
        \centering
    \includegraphics[width=\textwidth]{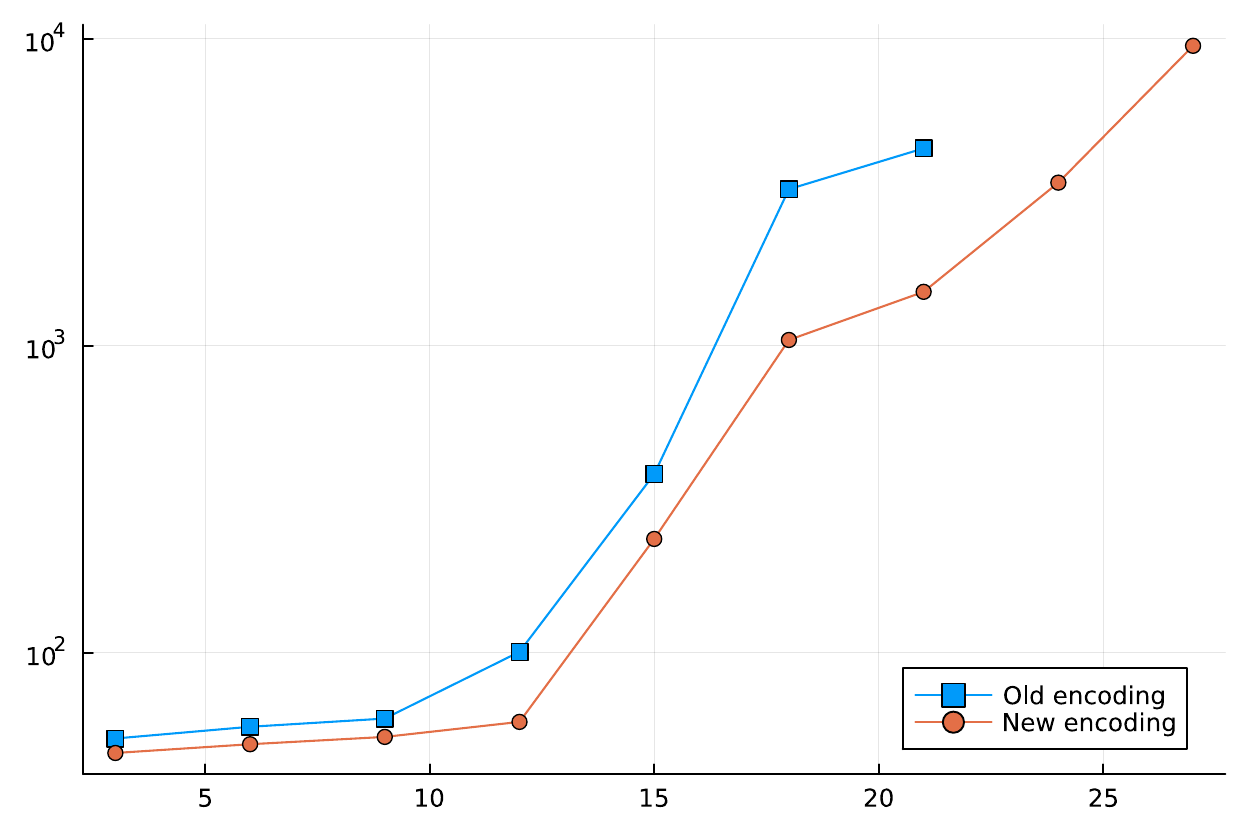}
    \caption{$2$-regular digraphs}
    \end{subfigure}
    \hfill
    \begin{subfigure}{0.45\linewidth}
        \centering
        \includegraphics[width=\textwidth]{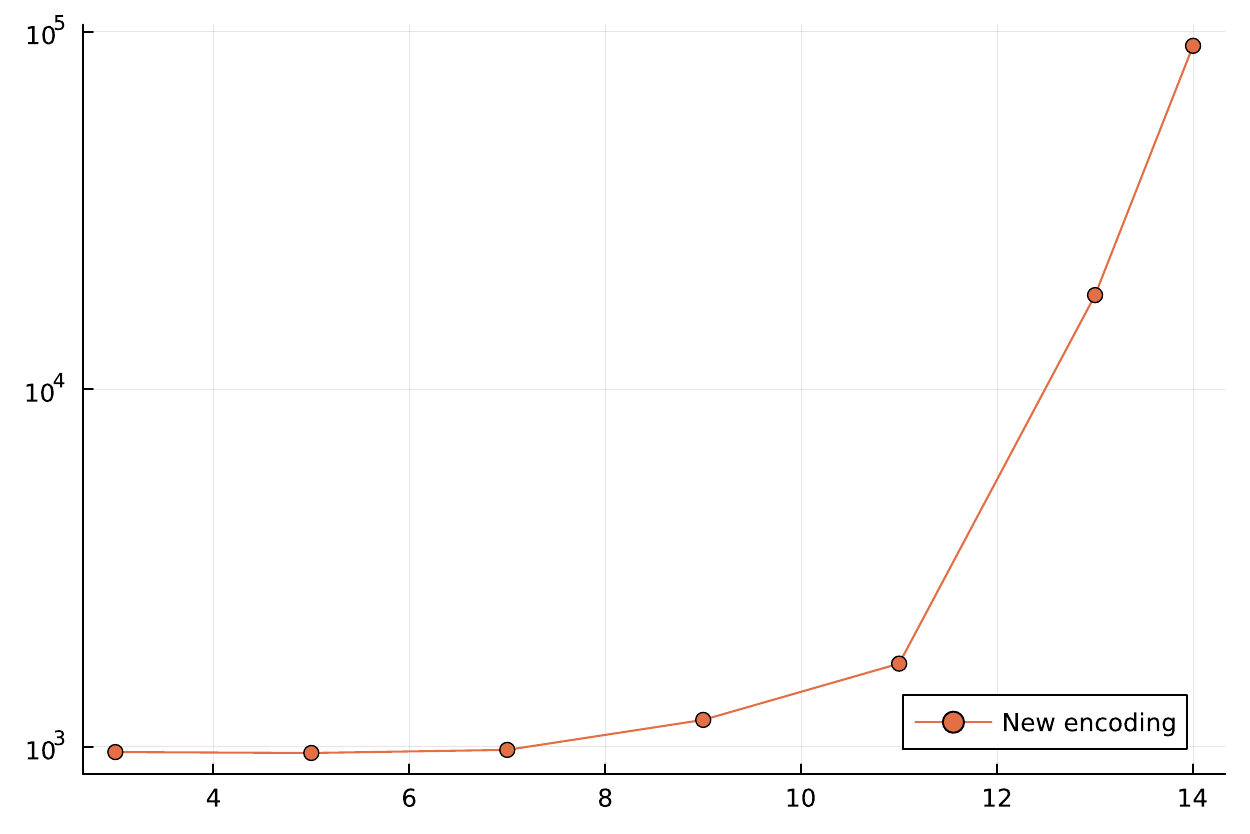}
        \caption{$3$-regular digraphs}
    \end{subfigure}
    \hfill
    \caption{Runtime for counting $k$-regular digraphs}
    \label{fig:3reg-digraph}
\end{figure*}

\subsubsection{Counting with the Linear Order Axiom}
Next, let us consider a sentence from the language of \ctwo{} extended with the linear order axiom.
Sentence $\Gamma_4$ encodes a graph similar to the Model A of the Barab\'{a}si-Albert model \citep{barabasi-laszlo02:random-graph}, an algorithm for generating random networks:
\begin{align*}
    \Gamma_4 &= \left(\forall x\; Eq(x, x)\right) \wedge (|Eq|=n)\\
    &\wedge \exists^{=k+1} x\; K(x)\\
    &\wedge \forall x\; \neg R(x, x)\\
    &\wedge \forall x \forall y\; K(x) \wedge K(y) \wedge \neg Eq(x, y) \Rightarrow R(x, y)\\
    &\wedge \forall x \exists^{=k} y\; R(x, y)\\
    &\wedge \forall x \forall y\; R(x, y) \wedge \neg \left(K(x) \wedge K(y)\right) \Rightarrow y \leq x\\
    &\wedge \forall x \forall y\; K(x) \wedge \neg K(y) \Rightarrow x \leq y\\
    &\wedge Linear(\leq)
\end{align*}

In a sense, the graph encoded by $\Gamma_4$ on $n$ vertices is sequentially \emph{grown}.
We start by ordering the vertices using the linear order axiom.
Then, a complete graph $K_{k+1}$ is formed on the first $k + 1$ vertices.
Afterward, we start growing the graph by \emph{appending} remaining vertices $i\in\{k+2, k+3,\ldots,n\}$ one at a time.
When appending a vertex $i$, we introduce $k$ outgoing edges that can only connect to the vertices $\{1, 2, \ldots, i - 1\}$, i.e., all the new edges have a form $(i, j)$ where $j\in\{1, 2, \ldots, i - 1\}$.
Ultimately, when counting the number of such graphs, we may not be interested in the same solutions differing by vertex ordering only, so we can divide the final number by $n!$.
The ordering through the linear order axiom is, however, a very useful modeling construct---without it, modeling graphs such as the one above would likely not be possible in a domain-lifted way.

From now on, let us refer to graphs defined by the sentence $\Gamma_4$ as BA$(k)$.
Figure~\ref{fig:ba3-counting} depicts the runtime for counting the graphs BA(3) on $n$ vertices using the old and the new encoding. The problems lead to 16 and 8 valid cells, respectively.

\begin{figure}[tb]
    \centering
    \begin{subfigure}{0.45\linewidth}
        \centering
        \includegraphics[width=\textwidth]{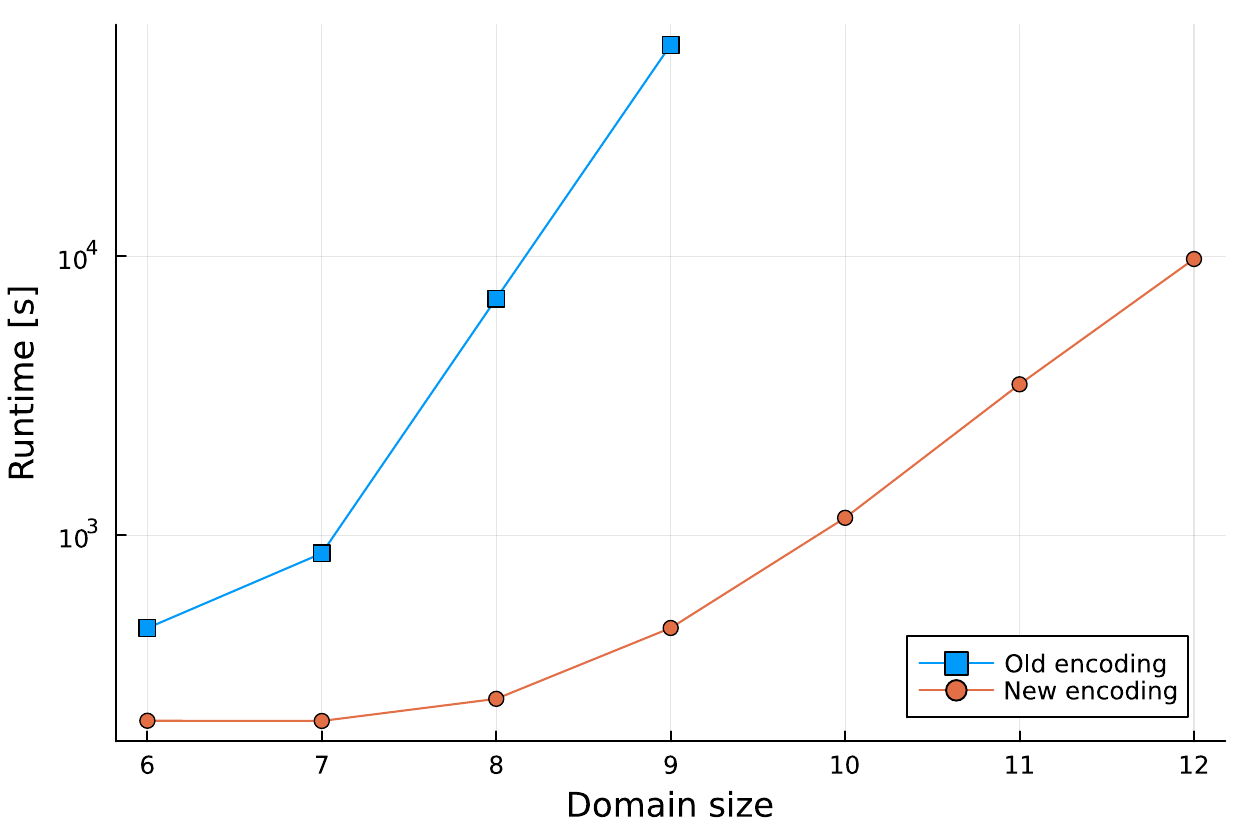}
        \caption{Counting BA(3) graphs}
        \label{fig:ba3-counting}
    \end{subfigure}
    \begin{subfigure}{0.45\linewidth}
        \centering
        \includegraphics[width=\textwidth]{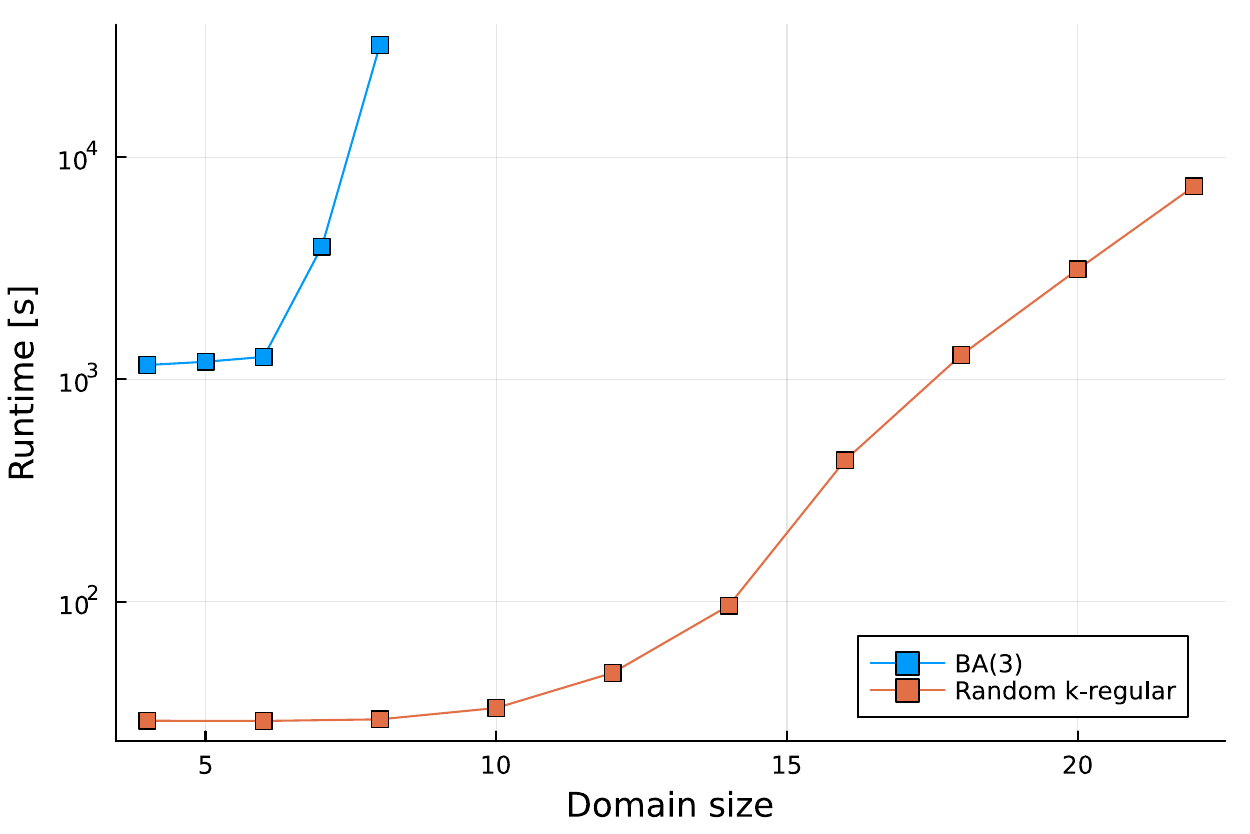}
        \caption{Inference on an MLN over BA(3) network}
        \label{fig:ba3-inference}
    \end{subfigure}
    \caption{Runtimes of tasks performed on BA(3)}
    \label{fig:ba3}
\end{figure}

\newpage
\subsection{Remarks on Performance}
The experiments above confirm that the new encoding is indeed more efficient, outperforming the old encoding on all tested instances.
While that should not come as a surprise, since we have derived a bound provably better than the old one, it is also not completely obvious because FastWFOMC uses many algorithmic tricks that partially make up for the inefficiencies of the old encoding.
However, not even our new bound allowed us to compute WFOMC within reasonable time for large domain sizes on all the tested problems.

Still, one should keep in mind that the alternative, that is, solving WFOMC by propositionalization to WMC scales extremely poorly, as was repeatedly shown in the lifted inference literature \citep{meert-etal14:wfomc-vs-wmc}.
Thus, we must rely on algorithms operating at the lifted level. Our work extends the domains that can be efficiently handled by lifted algorithms, but more work is needed to extend the reach of lifted inference algorithms further.

\subsection{Lifted Inference in Markov Logic Networks}
Markov Logic Networks (MLNs) are a popular first-order probabilistic language based on undirected graphs \citep{richardson-domingos06:mln}.
An MLN $\Phi$ is a set of weighted first-order logic formulas (possibly with free variables) with weights taking on values from the real domain or infinity:
$$\Phi = \{(w_1, \alpha_1), (w_2, \alpha_2), \ldots, (w_m, \alpha_m)\}$$
Given a domain $\Delta$, the MLN defines a probability distribution over possible worlds such as
\begin{align*}
    Pr_{\Phi, \Delta}(\omega) = \frac{\llbracket \omega \models \Phi_{\infty} \rrbracket}{Z} \exp\left(\sum_{(w_i, \alpha_i) \in \Phi_{\Real}} w_i \cdot N(\alpha_i, \omega) \right)
\end{align*}
where $\Phi_{\Real}$ denotes formulas with real-valued weights (soft constraints), $\Phi_{\infty}$ denotes formulas with infinity-valued weights (hard constraints), $\llbracket\cdot\rrbracket$ is the indicator function, $Z$ is the normalization constant ensuring valid probability values and $N(\alpha_i, \omega)$ is the number of substitutions to free variables of $\alpha_i$ that produce a grounding of those free variables that is satisfied in $\omega$.

Inference (and also learning) in MLNs is reducible to WFOMC \citep{broeck-etal14:wfomc-skolem}.
It is thus possible to efficiently perform exact lifted inference over an MLN using WFOMC as long as the network is encoded using a domain-liftable language.

The reduction works as follows:
For each $(w_i, \alpha_i(\mat{x}_i)) \in \Phi_\Real$, introduce a new formula $\forall \mat{x}_i: \xi_i(\mat{x}_i) \Leftrightarrow \alpha_i(\mat{x}_i)$, where $\xi_i$ is a fresh predicate, and set $w(\xi_i) = \exp(w_i), \overline{w}(\xi_i) = 1$ and $w(Q) =  \overline{w}(Q) = 1$ for all other predicates $Q$.
Formulas in $\Phi_{\infty}$ are added to the theory as additional constraints.
Denoting the new theory by $\Gamma$ and a query by $\phi$, we can compute the inference as
\begin{align*}
    Pr_{\Phi, \Delta}(\phi) = \frac{\textup{WFOMC}(\Gamma \wedge \phi, |\Delta|,\w)}{\textup{WFOMC}(\Gamma, |\Delta|, \w)}.
\end{align*}

We consider an MLN defined on the network BA($k$), i.e., encoded using the sentence $\Gamma_4$.
We extend the sentence with the predicate $Fr/2$ encoding undirected edges and we consider the standard MLN scenario where a smoker's friend is also a smoker.
That property, however, does not hold universally, instead, it is weighed, determining how important it is for an interpretation to satisfy it.
Putting it all together, we obtain an MLN in $\Phi_1$, namely
\begin{align*}
    \Phi_1 = \{ &\left(\infty, \forall x\; Eq(x, x) \right),\\
        &\left(\infty, (|Eq|=n) \right),\\
        &\left(\infty, \exists^{=k+1} x\; K(x) \right),\\
        &\left(\infty, \forall x\; \neg R(x, x) \right),\\
        &\left(\infty, \forall x \forall y\; K(x) \wedge K(y) \wedge \neg Eq(x, y) \Rightarrow R(x, y) \right),\\
        &\left(\infty, \forall x \exists^{=k} y\; R(x, y) \right),\\
        &\left(\infty, \forall x \forall y\; R(x, y) \wedge \neg \left(K(x) \wedge K(y)\right) \Rightarrow y \leq x \right),\\
        &\left(\infty, \forall x \forall y\; K(x) \wedge \neg K(y) \Rightarrow x \leq y \right),\\
        &\left(\infty, Linear(\leq) \right),\\
        &\left(\infty, \forall x \forall y\; Fr(x, y) \Leftrightarrow \left(R(x, y) \vee R(y, x)\right)\right),\\
        &\left(\ln w, Sm(x) \wedge Fr(x, y) \Rightarrow Sm(y) \right) \}
\end{align*}

We will be interested in the probability of there being exactly $m$ smokers.
To put the result into some context, we will compare them to a situation, where the underlying network is a random $k$-regular graph, i.e.,
\begin{align*}
    \Phi_2 = \{ &\left(\infty, \forall x\; \neg Fr(x, x) \right),\\ 
        &\left(\infty, \forall x \forall y\; Fr(x, y) \Rightarrow Fr(y, x)\right),\\
        &\left(\infty, \forall x \exists^{=k} y\; Fr(x, y) \right),\\
        &\left(\ln w, Sm(x) \wedge Fr(x, y) \Rightarrow Sm(y) \right) \}.
\end{align*}

Figure~\ref{fig:mln} shows the inference results for $n = 8$ and various weights $w$.
For both graphs, the distribution starts as binomial for small $w$.
Obviously, if we do not care about satisfying the smoker-friend property, random assignment of truth values to $Sm(t_1), Sm(t_2), \ldots, Sm(t_n)$ will lead to a binomial distribution.
As we increase the weight, the distributions gradually \emph{reverse} to the opposite scenario, a U-shaped distribution that prefers the extremes.
That is also intuitive, as we mostly work with one connected component where everyone is everyone's friend.
However, a random $3$-regular graph appears to go through the uniform distribution whereas our BA model seems not to.
That is due to the fact that the BA model is, by construction, always one connected component as opposed to a 3-regular graph that can also consist of two connected components on $n=8$ vertices.
Runtimes of the inference task for several domain sizes can be inspected in Figure~\ref{fig:ba3-inference}.

\vfill

\begin{figure*}[ht]
    \hfill
     \begin{subfigure}[b]{0.45\linewidth}
         \centering
         \includegraphics[width=\textwidth]{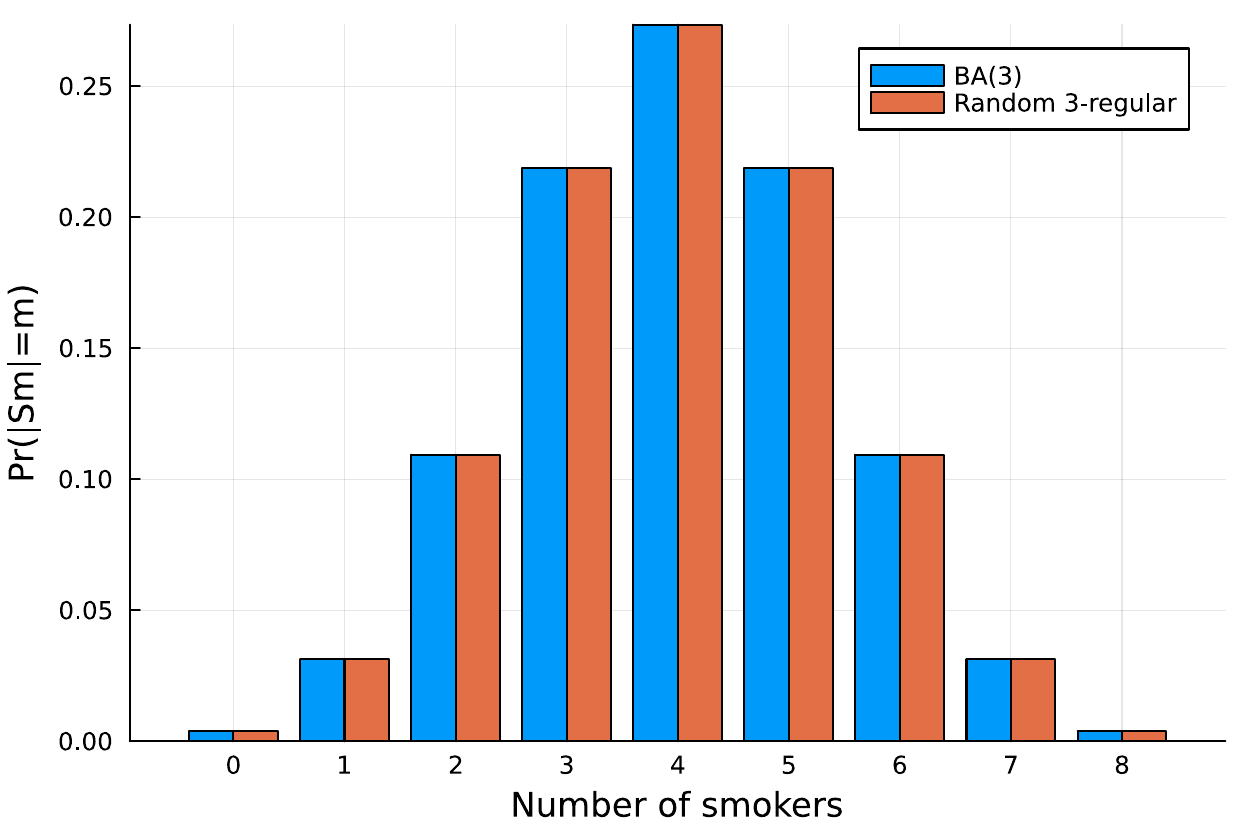}
         \caption{$w=1$}
         \label{fig:mln-1}
     \end{subfigure}
     \hfill
     \begin{subfigure}[b]{0.45\linewidth}
         \centering
         \includegraphics[width=\textwidth]{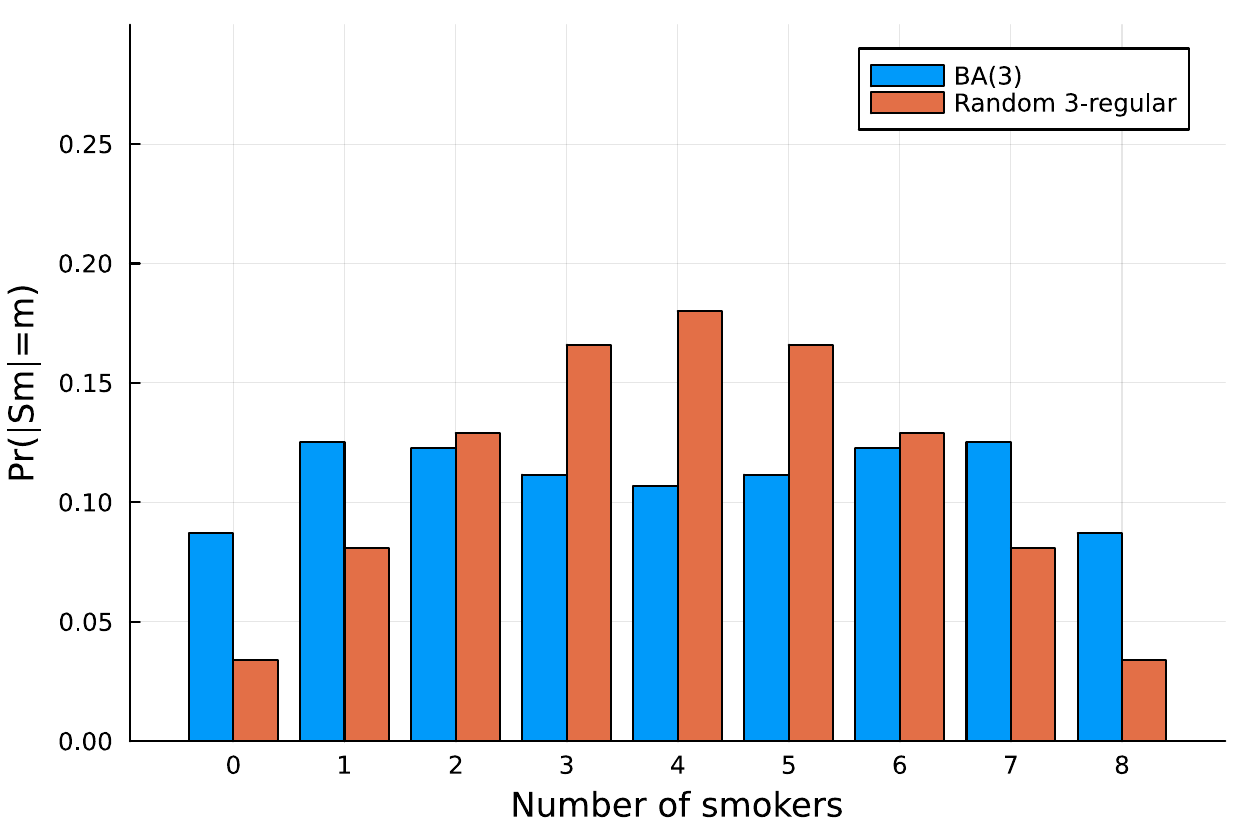}
         \caption{$w = \frac{3}{2}$}
         \label{fig:mln-3/2}
     \end{subfigure}
     \hfill

     \hfill
     \begin{subfigure}[b]{0.45\linewidth}
         \centering
         \includegraphics[width=\textwidth]{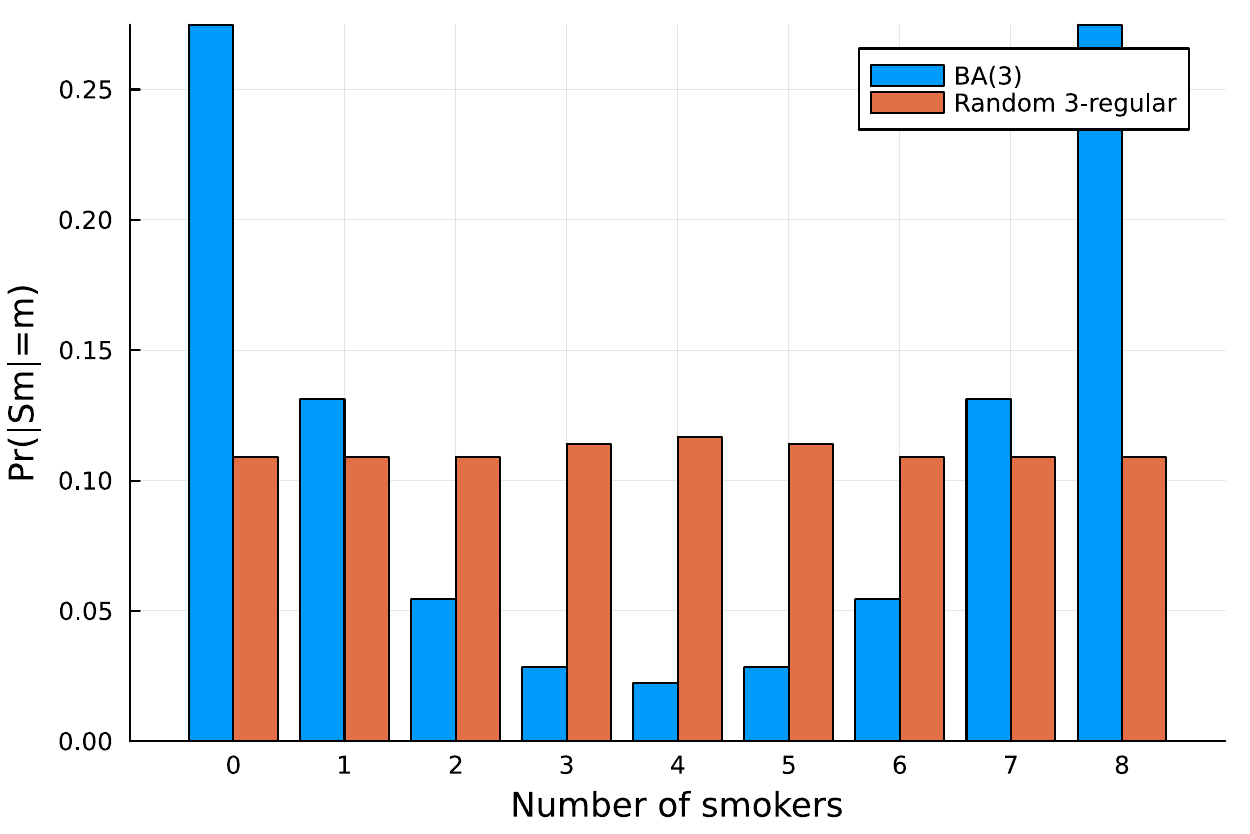}
         \caption{$w = 2$}
         \label{fig:mln-2}
     \end{subfigure}
     \hfill
     \begin{subfigure}[b]{0.45\linewidth}
         \centering
         \includegraphics[width=\textwidth]{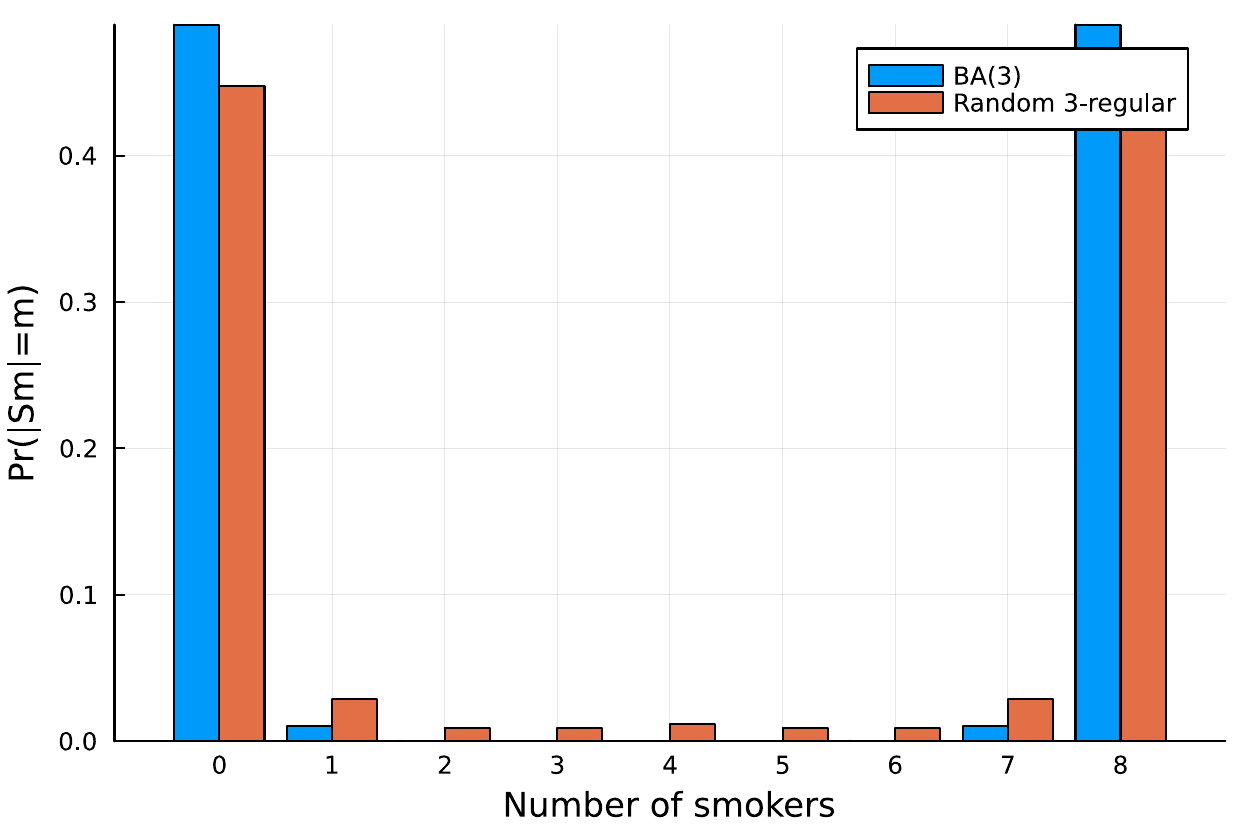}
         \caption{$w = 5$}
         \label{fig:mln-5}
     \end{subfigure}
    \hfill
    
     \caption{Probability of $m$ smokers}
    \label{fig:mln}
\end{figure*}

\vfill

\newpage

\section{Related Work}
\label{sec:other-works}
This work builds on a long stream of results from the area of lifted inference \citep{poole03:fove,braz05:fove2,jha10:lifted-inference-features,gogate11:ptp,taghipour13:fove3,braun16:lifted-junction-tree,dilkas23}.
Particularly, we continue in the line of research into the task of weighted first-order model counting \citep{broeck-etal11:knowledge-compilation,broeck11:domain-liftability-fo2,broeck-etal14:wfomc-skolem,beame-etal15:fo3-intractable,kazemi-etal16:sfo2+sru,kuusisto-lutz18:function-constraint,kuzelka21:wfomc-in-c2,bremen-kuzelka21b:tree-axiom,malhotra-serafini22:c2-closed-form,toth-kuzelka23:linear-order,malhotra-serafini23:dag-axiom,malhotra-etal23:graph-axioms}.

To the best of our knowledge, there is no other literature available on the exact complexity of computing WFOMC for \ctwo{} sentences.
The closest resource is the one proving \ctwo{} to be domain-liftable \citep{kuzelka21:wfomc-in-c2}, which we directly build upon and, in some sense, extend.
Besides that, \citet{malhotra-serafini22:c2-closed-form} later proposed a slightly different approach to dealing with counting quantifiers, although they did not analyze the method's exact complexity either.
However, as shown in the appendix, their techniques are also super-exponential in the counting parameters, not offering any speedup.
Another relevant resource, concerned with designing an efficient algorithm for computing WFOMC over \fotwo{}, is \citet{bremen-kuzelka21:fast-wfomc} whose FastWFOMC algorithm remains state-of-the-art and it can be used as a WFOMC oracle required to deal with cardinality constraints.

\section{Conclusion}
\label{sec:conclusion}
The best existing bound for the time complexity of computing WFOMC over \ctwo{} is polynomial in the domain size \citep{kuzelka21:wfomc-in-c2}.
However, as we point out, the polynomial's degree is exponential in the parameter $k$ of the counting quantifiers.
Using the new techniques presented in this paper, we reduce the dependency of the degree on $k$ to a quadratic one, thus achieving a super-exponential speedup of the WFOMC runtime with respect to the counting parameter $k$.

The new encoding can potentially improve any applications of WFOMC over \ctwo{} or make some applications tractable in the practical sense.
We support this statement further in the experimental section, where we provide runtime measurements for computing WFOMC of several \ctwo{} sentences and sentences from a domain-liftable \ctwo{} extension.

It remains an open question whether the complexity can be reduced even further.
Thus, we only consider our new bound a bound to beat, and we certainly hope that someone will beat it in the future.

\section*{Acknowledgments}
This work has received funding from the European Union’s Horizon Europe Research and Innovation
program under the grant agreement TUPLES No 101070149.
JT's work was also supported by a CTU grant no. SGS23/184/OHK3/3T/13.

\newpage
\begin{appendices}
\section{Bit Complexity}
All of the bounds presented in this paper assume mathematical operations to take constant time, which is a common assumption when it comes to asymptotic complexity.
However, in the context of WFOMC, it is quite common for the algorithm to end up working with extremely large values.

Consider counting the number of undirected graphs without loops on $n$ vertices, i.e., computing WFOMC of a formula
\begin{align*}
    \varphi = \left(\forall x\; \neg E(x, x)\right) \wedge \left(\forall x \forall y\; E(x, y) \Rightarrow E(y, x)\right)
\end{align*}
with $w(E)=\overline{w}(E)=1$.
Then
\begin{align*}
    \textup{WFOMC}(\varphi, n, \w) = 2^{\binom{n}{2}}
\end{align*}
and for $n = 12$, we already obtain
$$\textup{WFOMC}(\varphi, n, \w)=73786976294838206464,$$
i.e., a value requiring more than 64 bits for its machine representation,\footnote{The formula along with the value for $n=12$ can be checked at \url{https://oeis.org/A006125}.} leading to situations where bit complexity should also be considered.

In \citet{kuzelka21:wfomc-in-c2}, the author derives a bound on the number of bits required to represent each summand in the WFOMC computation, which turns out to be polynomial in $n$ and $\log M$, where $M$ is a bound on the number of bits required to represent the weights.
The total number of such summands is then also polynomial in $n$.
Thus, the polynomial bound for domain-liftability is not violated.

The derivation itself is rather involved.
We refer interested readers to Proposition 2 and its proof in \citet{kuzelka21:wfomc-in-c2}.
For our purposes, it suffices to say that the derivation would remain unchanged under our WFOMC transformations.

\section{Complexity of an Alternative Approach to Counting Quantifiers}
In this work, we focus on and improve techniques proposed in \citet{kuzelka21:wfomc-in-c2} for dealing with counting quantifiers when computing WFOMC.
A slightly different approach to counting in WFOMC was later proposed in \citet{malhotra-serafini22:c2-closed-form}.
For completeness, this section contains a proof sketch of their techniques also being super-exponential in the counting parameters, same as the approach from \citet{kuzelka21:wfomc-in-c2}.

Let us consider a concrete example of computing WFOMC over a formula encoding $m$-regular graphs,%
\footnote{We denote the degree by $m$ in this case to easily differentiate it from a vector $\mathbf{k}$ used in \citet{malhotra-serafini22:c2-closed-form}.}
i.e.,
\begin{align*}
    \Gamma &= \left(\forall x\; \neg E(x, x)\right) \wedge \left(\forall x \forall y\; E(x, y) \Rightarrow E(y, x)\right)\\
    &\wedge \left(\forall x \exists^{=m} y\; E(x, y)\right).
\end{align*}
Following Theorem~3 in \citet{malhotra-serafini22:c2-closed-form}, when counting over $\Gamma$, we first 
 replace the counting subformula by an atom $A(x)$ with a fresh predicate $A/1$, i.e.,
\begin{align*}
    \Gamma' &= \left(\forall x\; \neg E(x, x)\right) \wedge \left(\forall x \forall y\; E(x, y) \Rightarrow E(y, x)\right)\\ 
    &\wedge \left(\forall x\; A(x) \right) \wedge \left(\forall x\; A(x) \Leftrightarrow \left(\exists^{=m} y\; E(x, y)\right)\right).
\end{align*}

Second, we define additional formulas
\begin{align*}
    \Phi_1 &= \bigwedge_{i=1}^m \forall x \forall y\; P_i(x) \Rightarrow \neg \left(\left(A(x) \vee B(x)\right) \Rightarrow f_i(x, y)\right)\\
    \Phi_2 &= \bigwedge_{1 \leq i < j \leq m} \forall x \forall y\; f_i(x, y) \Rightarrow \neg f_j(x, y)\\
    \Phi_3 &= \bigwedge_{i=1}^m \forall x \forall y\; f_i(x, y) \Rightarrow E(x, y)\\
    \Phi_4 &= \forall x\; B(x) \Rightarrow \neg A(x)\\
    \Phi_5 &= \forall x \forall y\; M(x, y) \Leftrightarrow \left(\left(A(x) \vee B(x)\right) \wedge E(x, y)\right)\\
    \Phi_6 &= |A| + |B| = |f_1| = |f_2| = \ldots = |f_m| = \frac{|M|}{m},
\end{align*}
where $B, M, f_i$'s and $P_i$'s are fresh predicates.

Third, we eliminate the counting sentence in $\Gamma'$ using the new formulas as
\begin{align*}
    \Gamma'' &= \left(\forall x\; \neg E(x, x)\right) \wedge \left(\forall x \forall y\; E(x, y) \Rightarrow E(y, x)\right)\\ 
    &\wedge \left(\forall x\; A(x) \right) \wedge \bigwedge_{i=1}^5 \Phi_i.
\end{align*}

Finally, the counting problem can be posed as
\begin{align}
\label{eq:malhotra}
    \textup{WFOMC}(\Gamma, n, \w) = \sum_{(\mathbf{k}, \mathbf{h}) \models \Phi_6} \binom{n}{\mathbf{k}} \; \mathcal{T}_m(\Gamma'', \mathbf{k}, \mathbf{h}, \w),
\end{align}
where $\binom{n}{\mathbf{k}}$ is the multinomial coefficient, $\mathcal{T}_m(\Gamma'', \mathbf{k}, \mathbf{h}, \w)$ is some function whose complexity we can assume to be constant with respect to $n$, and the summation argument $(\mathbf{k}, \mathbf{h}) \models \Phi_6$ denotes all vector pairs $(\mathbf{k}, \mathbf{h})$ filtered so that they satisfy the cardinality constraints in $\Phi_6$.
We leave to \citet{malhotra-serafini22:c2-closed-form} what exactly it means for $(\mathbf{k}, \mathbf{h})$ to satisfy $\Phi_6$.
For our analysis, it suffices to say that vector $\mathbf{k}$ orders cells of $\Gamma''$ and $\mathbf{k}_i$ specifies how many constants from the domain satisfy the $i$-th cell in an interpretation.
Note that, in our example, the filtering will not reduce the number of instances of vector $\mathbf{k}$ that needs to be summed over.
Hence, the time complexity of evaluating Equation~\ref{eq:malhotra} is at least exponential in the dimensionality of $\mathbf{k}$.


Let us now inspect the structure of cells of $\Gamma''$ which will allow us to compute their total number, i.e., it will allow us to exactly determine the dimensionality of vector $\mathbf{k}$.
All cells must contain $\neg E(x, x)$ and $A(x)$ due to the non-counting conjuncts in $\Gamma'$.
Due to $\Phi_4$, cells must further contain $\neg B(x)$ and, by $\Phi_5$, $\neg M(x, x)$ as well.
Furthermore, due to $\Phi_3$, the cells must also contain $\neg f_i(x, x)$ for all $i$.
Atoms enforced so far cause the consequents of all implications in $\Phi_1$ to be satisfied, hence, we are free to choose the truth values of all $P_i(x)$ atoms.
Therefore, we have $2^m$ cells and the dimensionality of $\mathbf{k}$ is exponential in the counting parameter $m$.
Overall, we conclude that the approach from \citet{malhotra-serafini22:c2-closed-form} for computing WFOMC over \ctwo{} is, in fact, also super-exponential in the counting parameters.

\end{appendices}

\newpage

\bibliographystyle{named}
\bibliography{main}

\end{document}